%% file: Main.tex
\documentclass[runningheads,a4paper, envcountsame, envcountsect]{llncs}

\usepackage{graphicx}

\usepackage{amssymb, amsmath}
\usepackage{color}
\usepackage{graphicx}
\usepackage{algorithm}
\usepackage[noend]{algpseudocode}
\usepackage{cleveref}
\usepackage{authblk}
\usepackage{times}
\usepackage{setspace}
\usepackage{cite}

\newcommand{\polylogn}{\text{polylog}(n)}

\newcommand{\Delt}{\Delta}

\newcommand{\inQuote}[1]{: ``#1''}
\newcommand{\FullOrShort}{full}

\ifthenelse{\equal{\FullOrShort}{full}}{
	  
	  \newcommand{\fullOnly}[1]{#1}
	  \newcommand{\shortOnly}[1]{}

  }{
   \newcommand{\fullOnly}[1]{}
	 \newcommand{\shortOnly}[1]{#1}
  }

\usepackage{url}

\begin{document}

\mainmatter  

\title{Bounds on Contention Management in Radio Networks}

\titlerunning{Bounds on Contention Management in Radio Networks}

%
%

  \author[*]{Mohsen Ghaffari}
	\author[*]{Bernhard Haeupler}
	\author[*]{Nancy Lynch}
	\author[**]{Calvin Newport\vspace{0.4cm}}
	\affil[*]{Computer Science and Artificial Intelligence Lab, MIT \authorcr \texttt{\{ghaffari, haeupler, lynch\}@csail.mit.edu}\vspace{0.2cm}}
	\affil[**]{Department of Computer Science, Georgetown University \authorcr \texttt{cnewport@cs.georgetown.edu}}

\renewcommand\Authands{ and }

\authorrunning{Ghaffari et al.}
\institute{}

%
%

\toctitle{Bounds on Contention Management in Radio Networks }
\tocauthor{Ghaffari et al.}
\maketitle

\fullOnly{\thispagestyle{empty}}

\begin{abstract} 
The local broadcast problem assumes that processes in a 
wireless network are provided messages, one by one,
that must be delivered to their neighbors.
In this paper, 
we prove tight bounds for this problem in two well-studied
wireless network models:
the {\em classical} model,
in which links are reliable and collisions consistent,
and the more recent {\em dual graph} model,
which introduces unreliable edges.
Our results prove that the
{\em Decay} strategy,
commonly used for local broadcast in the classical setting,
is  optimal.
They also establish a separation between the two models,
proving that the dual graph setting is strictly harder
than the classical setting, with respect to this primitive.
%
\end{abstract}

\input{Preliminaries}

\input{RelatedWork}
\input{UpperBounds}
\input{Classic}

\input{Dual}
\input{Gap}
\vspace{-0.2cm}
\input{Ref-List}

\end{document}

%% file: Preliminaries.tex
\section{Introduction} 
At the core of every wireless network algorithm is the need
to manage contention on the shared medium.
In the theory community, this challenge is abstracted
as the {\em local broadcast problem}, in which processes
are given messages, one by one, 
that must be delivered to their neighbors.

This problem has been studied in multiple wireless network models.
The most common such model is the {\em classical} model, introduced
by Chlamatac and Kutten~\cite{CK85}, in which
links are reliable
and concurrent broadcasts by neighbors always generate collisions.
The dominant local broadcast strategy in this model
is the {\em Decay} routine introduced by Bar-Yehuda et al.~\cite{BGI87}.
In this strategy,
nodes cycle through an exponential distribution of broadcast
probabilities
with the hope that one will be appropriate for the current 
level of contention~(e.g., 
\cite{BGI87, CGR00, CGGPR00, CGOR00, CMS01, CCMPS01, CMS04, GPX05, KLNOR10}).
To solve local broadcast with high probability (with respect to the network size $n$),
the {\em Decay} strategy requires $O(\Delta\log{n})$ rounds,
where $\Delta$ is the maximum contention in the network (which is at most the maximum degree in the network topology).
It has remained an open question whether this bound can be improved
to $O(\Delta + \polylogn)$. 
In this paper, we resolve this open question by proving the {\em Decay}
bound optimal.
This result also proves that existing
constructions of {\em ad hoc selective families}~\cite{CCMPS01,CMS04}---a 
type
of combinatorial object used in wireless network 
algorithms---are optimal.

We then turn our attention to the more recent {\em dual graph}
wireless network model introduced by Kuhn et~al.~\cite{KLN09, KLN09BA, KLNOR10, CGKLN11}.
This model generalizes the classical model by allowing
some edges in the communication graph to be unreliable.
It was motivated by the observation 
that real wireless networks include links of dynamic quality (see~\cite{KLNOR10} for more
extensive discussion).
We provide tight solutions to the local broadcast problem in this setting, using algorithms
based on the {\em Decay} strategy.
Our tight bounds in the dual graph model are larger (worse) 
than our tight time bounds for the classical model, formalizing
a separation between the two settings (see Figure~\ref{fig:results} and
the discussion below for result details).
We conclude by proving another separation:
in the classical model there is no significant difference
in power between centralized and distributed local broadcast
algorithms, while in the dual graph model the gap is exponential.

These separation results are important 
because most wireless network algorithm
analysis relies on the correctness of the underlying contention
management strategy.
By proving that the dual graph model is strictly harder
with respect to local broadcast,
we have established that an algorithm
proved correct in the classical model will not necessarily
remain correct or might loose its efficiency in the more general (and more realistic) 
dual graph model.
%


\underline{To summarize:} 
This paper provides an essentially complete characterization
of the local broadcast problem in the well-studied classical and dual graph wireless network models. 
In doing so, we: 
(1) answer the long-standing open question regarding the optimality of {\em Decay} in the classical model;
(2) provide a variant of Decay and prove it optimal for the local broadcast problem in the dual graph model; and
(3) formalize the separation between these two models, with respect to local broadcast.
%
%
{\begin{figure}[t]
\centering
\shortOnly{\tiny}
\begin{tabular}{|c|c|c|}
\hline
  & \bf Classical Model & \bf Dual Graph Model \\ \hline \hline
\bf Ack. Upper  & $O(\Delta\log{n})$** & $O(\Delta'\log{n})$* \\ \hline
\bf Ack. Lower  & $\Omega(\Delta \log{n})$* 
  & $\Omega(\Delta'\log{n})$*\\ \hline \hline
\bf Prog. Upper & $O(\log{\Delta}\log{n})$** & $O(\min\{ k\log{k}\log{n}, \Delta'\log{n} \})$* \\ \hline
\bf Prog. Lower & $\Omega(\log{\Delta}\log{n})$** & $\Omega(\Delta'\log{n})$*\\ \hline
\end{tabular}
\caption{{\shortOnly{\footnotesize} \onehalfspacing A summary of our results for {\em acknowledgment} and {\em progress} 
for the local broadcast problem.
Results that are new, or significant improvements
over the previously best known result, are marked
with an ``*'' while a ``**'' marks results that where obtained from 
prior work via minor tweaks.
}}
\label{fig:results}
\end{figure}}

\paragraph{Result Details:}
As mentioned, 
the {\em local broadcast} problem assumes processes are
provided messages, one by one, which should be delivered
to their neighbors in the communication graph.
Increasingly, local broadcast solutions
are being studied separately from the higher level problems
that use them, 
improving the composability of 
solutions; e.g.,~\cite{KLN09,CLVW09,KKKL10,KKLMP11}.
Much of the older theory work in the wireless setting, however,
mixes the local broadcast logic with the logic
of the higher-level problem being solved; e.g.,~\cite{BGI87, CGR00, CGGPR00, CGOR00, CMS01, CCMPS01, CMS04, GPX05, KLNOR10}.
This previous work can be seen as implicitly solving local broadcast.

The efficiency of a local broadcast algorithm is characterized
by two metrics: 
(1) an {\em acknowledgment bound}, which measures the time for 
a sender process (a process that has a message for broadcast) 
to deliver its message to all of its neighbors; 
and (2) a {\em progress bound}, which measures
the time for a receiver process (a process that has a sender neighbor) 
to receive at least one message~\footnote{Note that with respect to these definitions, a process can be both a sender and a receiver, simultaneously.}.
The acknowledgment bound is obviously interesting; 
the progress bound has also been shown to be critical for analyzing 
algorithms for many problems, e.g., global broadcast~\cite{KLN09} 
where the reception of {\em any} message is normally sufficient to advance the algorithm.
The progress bound was first introduced and explicitly 
specified in~\cite{KLN09, KKKL10}
but it was implicitly used already in many previous works~\cite{BGI87, CGR00, CGGPR00, CGOR00, CMS01, GPX05}.
Both acknowledgment and progress bounds typically depend on two parameters, the maximum contention $\Delta$
and the network size $n$. In the dual graph model,
an additional measure of maximum contention, $\Delta'$, is introduced
to measure contention in the unreliable communication link
graph, which is typically denser than the reliable link graph.
In our progress result for the dual graph model, we also
introduce $k$ to capture the {\em actual} amount of contention relevant
to a specific message.
These bounds are usually required to hold with high probability. 

Our upper and lower bound results for the
local broadcast problem 
in the classical and dual graph models are summarized
in Figure~\ref{fig:results}.
Here we highlight three key points regarding these results. 
First, in both models, the upper and lower bounds match asymptotically.
Second, we show that $\Omega(\Delta\log{n})$ rounds
are necessary for acknowledgment in the classical model.
This answers in the negative the open question of whether
a $O(\Delta + \polylogn)$ solution is possible. 
Third, the separation between the classical and dual graph
models occurs with respect to the progress bound,
where the tight bound for the classical model
is {\em logarithmic} with respect to contention,
while in the dual graph model it is {\em linear}---an exponential
gap.
%
%
Finally, in addition to the results described
in Figure~\ref{fig:results}, we also prove the following
additional separation between the two models:
in the dual graph model, the gap in progress
between distributed and centralized local
broadcast algorithms is (at least) linear in the
maximum contention $\Delta'$, whereas no such gap exists in the classical
model.

Before starting the technical sections, we remark that due to space considerations, the full proofs are omitted from the conference version and can be found in~\cite{GHLN12}.

\input{model}

\input{Problem}

%% file: model.tex
\section{Model}
\label{sec:model}
To study the local broadcast problem in synchronous multi-hop radio networks, we use two models, namely the \emph{classical radio network model} (also known as the radio network model) and the \emph{dual graph model}. The former model assumes that all connections in the network are reliable and it has been extensively studied since 1980s~\cite{CK85, BGI87, ABLP91, CGR00, CGGPR00, CGOR00, CMS01, CCMPS01, CMS04, GPX05, KLN09, KLN09, KKKL10}. On the other hand, the latter model is a more general model, introduced more recently in 2009~\cite{KLN09, KLN-DISC-09, KLN09BA}, which includes the possibility of unreliable edges. Since the former model is simply a special case of the latter, we use dual graph model for explaining the model and the problem statement. However, in places where we want to emphasize on a result in the classical model, we focus on the classical model and explain how the result specializes for this specific case.

In the dual graph model, radio networks have some reliable and 
potentially some unreliable links. Fix some $n\geq 1$. 
We define a network $(G,G')$
to consist of two undirected graphs, $G=(V,E)$
and $G'=(V,E')$,
where $V$ is a set of $n$ wireless nodes
and $E \subseteq E'$, where intuitively set $E$ is the set of reliable edges while $E'$ is the set of all edges, both reliable and unreliable. In the classical radio network model, there is no unreliable edge and thus, we simply have $G = G'$, i.e., $E = E'$.

We define an algorithm ${\cal A}$
to be a collection of $n$ randomized processes,
described by probabilistic automata.
An execution of ${\cal A}$ in network $(G,G')$
proceeds as follows:
first, we fix a bijection
$proc$ from $V$ to ${\cal A}$.
This bijection assigns
processes to graph nodes.
We assume this bijection is defined by
an adversary and is not known to the processes.
We do not, however, assume that the definition of $(G,G')$
is unknown to the processes (in many real world settings it
is reasonable to assume that devices can make some
assumptions about the structure of their network).
In this study, to strengthen our results,
our upper bounds make no assumptions
about $(G,G')$ beyond bounds on maximum contention and polynomial bounds on size of the network,
while our lower bounds allow full knowledge of the 
network graph.

An execution proceeds in synchronous rounds $1,2,...$,
with all processes starting in the first round.
At the beginning of each round $r$, every process
$proc(u), u\in V$
first receives inputs (if any) from the environment.
It then decides whether or not to transmit a message
and which message to send.
Next, the adversary chooses
a {\em reach set} that consists of $E$
and some subset, potentially empty,
of edges in $E'-E$. Note that in the classical model, set $E' - E$ is empty and therefore, the reach set is already determined. 
%
This set describes the links
that will behave reliably in this round. We assume that the adversary has full knowledge of the state of the network while choosing this reach set.
For a process $v$, let $B_{v,r}$ be the set all graph nodes $u$
such that, $proc(u)$ broadcasts in $r$ and $\{u,v\}$ is in the
reach set for this round.
What $proc(v)$ receives in this round \shortOnly{is determined as follows.}\fullOnly{depends
on the size of $B_{v,r}$, the messages
sent by processes assigned to nodes in $B_{v,r}$, 
and $proc(v)$'s behavior.}
If $proc(v)$ broadcasts in $r$, then it receives
only its own message. 
If $proc(v)$ does not broadcast, there are two cases:
(1) if $|B_{v,r}| = 0$ or $|B_{v,r}| > 1$, then $proc(v)$ 
receives $\bot$ (indicating {\em silence}); 
(2) if $|B_{v,r}| = 1$, then $proc(v)$ receives the message sent by
$proc(u)$, where $u$ is the single
node in $B_{v,r}$.
That is, we assume processes cannot send and receive simultaneously, 
and also, there is no collision detection in this model. However, to strengthen our results, we note that our lower bound results hold even in the model with collision detection, i.e., where process $v$ receives a special collision indicator message $\top$ in case $|B_{v,r}| > 1$. After processes receive their messages, they generate outputs (if any) to pass back to the environment.

\paragraph{Distributed vs. Centralized Algorithms:}
The model defined above describes distributed algorithms in a radio network setting. To strengthen our results, in some of our lower bounds we 
consider the stronger model of {\em centralized} algorithms.
%
We formally define a centralized algorithm to be defined
the same as the distributed algorithms above,
but with the following two modifications: (1) the processes
are given $proc$ at the beginning of the execution;
and (2) the processes can make use of the current state and inputs
of {\em all} processes in the network when making decisions
about their behavior.


\paragraph{Notation \& Assumptions:}
\fullOnly{The following notation and assumptions will simplify
the results to follow.}
For each $u\in V$, the notations
$N_G(u)$ and $N_{G'}(u)$ describe, respectively, the
neighbors of $u$ in $G$ and $G'$. Also, we define $N^{+}_{G}(u)= N_G(u) \cup \{u\}$ and $N^{+}_{G'}(u)= N_{G'}(u) \cup \{u\}$.
For any algorithm ${\cal A}$, 
we assume that each process ${\cal A}$
has a unique identifier. To simplify notation, we assume the identifiers are from $\{1,...,n\}$. We remark that our lower bounds hold even with such strong identifiers, whereas for the upper bounds, we just need the identifiers of different processes to be different.  
Let $id(u), u\in V$ describe the id of process $proc(u)$.
For simplicity,
throughout this paper we often use
the notation {\em process $u$}, or sometimes just $u$,
for some $u\in V$,
to refer to $proc(u)$ in the execution in question.
Similarly, we sometimes use {\em process $i$},
or sometimes just $i$,
for some $i\in\{1,...,n\}$,
to refer to the process with id $i$. We sometimes use the notation $[i,i']$, for integers $i' \geq i$, 
to indicate the sequence $\{i,...,i'\}$,
and the notation $[i]$ for integer $i$ to indicate $[1,i]$.
Throughout, we use the
the notation {\em w.h.p.} ({\em with high probability}) to indicate
a probability at least $1-\frac{1}{n}$. Also, unless specified, all logarithms are natural log. Moreover, we ignore the integral part signs whenever it is clear that omitting them does not effect the calculations more than a change in constants. 
 
%

%% file: Problem.tex
\section{Problem}
\label{sec:problem}
%
%

\fullOnly{\paragraph{Preliminaries:}}
Our first step in formalizing the local broadcast
problem is to fix the input/output interface between the \emph{local broadcast module} (automaton) of a process and the higher layers at that process. In this interface, there are three actions as follows: (1) $bcast(m)_v$, an input action that provides the local broadcast module at process $v$ with message $m$ that has to be broadcast over $v$'s local neighborhood, (2) $ack(m)_v$, an output action that the local broadcast module at $v$ performs to inform the higher layer that the message $m$ was delivered to all neighbors of $v$ successfully, (3) $rcv(m)_u$, an output action that local broadcast module at $u$ performs to transfer the message $m$, received through the radio channel, to higher layers. To simplify definitions going forward, we assume w.l.o.g. that
every $bcast(m)$ input in a given execution is for a unique $m$.
We also need to restrict the behavior of the environment
to generate $bcast$ inputs in a {\em well-formed} manner,
which we define as strict alternation between $bcast$
inputs and corresponding $ack$ outputs at each process.
In more detail, for every execution and every process $u$,
the environment generates a $bcast(m)_u$ input only 
under two conditions: (1) it is the first input to $u$
in the execution; or (2) the last input or non-$rcv$ output action
at $u$ was an $ack$.

\fullOnly{\paragraph{Local Broadcast Algorithm:}}
We say an algorithm {\em solves the local
broadcast problem} if and only if 
in every execution, we have the following three properties:
(1) for every process $u$, for each $bcast(m)_u$ input, $u$ eventually
responds with a single $ack(m)_u$ output, and
these are the only $ack$ outputs generated by $u$;
(2) for each process $v$, for each message $m$, $v$ outputs $rcv(m)_v$ at most once and if $v$ generates a $rcv(m)_v$ output in round
$r$, then there is a neighbor $u \in N_{G'}(v)$ such 
that following conditions hold: $u$ received a $bcast(m)_u$ input before round $r$ and has not output $ack(m)_u$ before round $r$
(3) for each process $u$, if $u$ receives $bcast(m)_u$ in round
$r$ and respond with $ack(m)_u$ in round
$r' \geq r$, then w.h.p.: $\forall v\in N_G(u)$,
$v$ generates output $rcv(m)_v$ 
within the round interval $[r,r']$.
We call an algorithm that solves the local broadcast
problem a {\em local broadcast algorithm}.

\paragraph{Time Bounds:}
We measure the performance of a local broadcast algorithm
with respect to the \fullOnly{three}\shortOnly{two} bounds first formalized
in~\cite{KLN09}: {\em acknowledgment} (the worst case bound
on the time between a $bcast(m)_u$ and the corresponding $ack(m)_u$),\fullOnly{{\em receive } (the worst case bound on the time between
a $bcast(m)_v$ input and a $rcv(m)_u$ output
for all $u\in N_{G}(v)$),}
and {\em progress} (informally speaking the worst case bound on the time 
for a process to receive at least one message 
when it has one or more $G$ neighbors with messages to send).
The first \fullOnly{two bounds represent}\shortOnly{bound represents} standard ways of measuring
the performance of local communication. The progress bound is crucial for obtaining tight performance
bounds in certain classes of applications. See \cite{KLN09, KKKL10} for examples of places where progress bound proves crucial explicitly. Also, \cite{BGI87, CGR00, CGGPR00, CGOR00, CMS01,  GPX05} use the progress bound implicitly throughout their analysis.

In more detail, a local broadcast algorithm has
\fullOnly{three}\shortOnly{two} {\em delay functions} which describe
these delay bounds as a function of the relevant contention: 
$f_{ack}$, \fullOnly{$f_{rcv}$, }and $f_{prog}$, respectively. 
In other words, every local broadcast algorithm can be characterized
by these \fullOnly{three} \shortOnly{two} functions which must satisfy
properties we define below.
Before getting to these properties, however,
we first present a few helper definitions that we use
to describe local contention during a given round interval.
The following are defined with respect to a fixed execution. 
(1) We say a process $u$ is {\em active} in round $r$,
or, alternatively, {\em active with $m$},
iff it received a $bcast(m)_u$ output in
a round $\leq r$ and it has not yet generated
an $ack(m)_u$ output in response. We furthermore
call a message $m$ active in round $r$ if there is a
process that is active with it in round $r$.  
(2) For process $u$ and round $r$, contention $c(u,r)$ equals
the number of active $G'$ neighbors of $u$ in $r$.
Similarly, for every $r' \geq r$, 
$c(u,r,r') = max_{r''\in[r,r']}\{c(u,r'')\}$. 
(3) For process $v$ and rounds $r' \geq r$, 
$c'(v,r,r') = max_{u\in N_G(v)}\{c(u,r,r')\}$. 
We can now formalize the properties our delay
functions, specified for a local broadcast algorithm, must satisfy
for any execution:
\begin{enumerate}
\fullOnly{\item {\em Receive bound:} Suppose that $v$ receives a $bcast(m)_v$
  input in round $r$ and $u\in N_{G'}(v)$ generates $rcv(m)_v$ in $r' \geq r$. Then with high probability we have $r' - r \leq f_{rcv}(c(u,r,r'))$.}

\item {\em Acknowledgment bound:} Suppose process $v$ receives a $bcast(m)_v$	input in round $r$. Then, if $r' \geq r$ is the round in which process $v$ generates corresponding output $ack(m)_v$,     
      then with high probability we have $r' - r \leq f_{ack}(c'(v,r,r'))$.

 \item {\em Progress bound:} For any pair of rounds $r$ and $r' \geq r$, and process $u$, if $r' - r > f_{prog}(c(u,r,r'))$ and there exists a neighbor $v\in N_G(u)$ that is active throughout the entire interval $[r,r']$, then with high probability, $u$ generates a $rcv(m)_u$ output in a round $r'' \leq r'$ for a message $m$ that was active at some round within $[r,r']$.

\end{enumerate}

We use notation $\Delta^\prime$ (or $\Delta$ for the classical model) to denote the maximum contention over all processes.\footnote{Note that since the maximum degree in the graph is an upper bound on the maximum contention, this notation is consistent with prior work, see e.g. ~\cite{KLN09, KKKL10, KKLMP11}.}  In our upper bound results, we assume that processes are provided with upper bounds on contention that are within a constant factor of $\Delta'$ (or $\Delta$ for the classical model). Also, for the sake of concision,
in the results that follow, 
we sometimes use the terminology ``{\em has an acknowledgment bound of}"
(resp.\fullOnly{{\em receive bound} and} {\em progress bound})
to indicate ``{\em specifies the delay function $f_{ack}$}" 
(resp. \fullOnly{$f_{rcv}$ and}$f_{prog}$).
For example,
instead of saying ``the algorithm specifies
delay function $f_{ack}(k) = O(k)$,'' we might
instead say ``the algorithm has an acknowledgment
bound of $O(k)$.''

\paragraph{Simplified One-Shot Setting for Lower Bounds:}
The local broadcast problem as just described assumes that processes can keep 
receiving messages as input forever and in an arbitrary asynchronous way.
This describes the practical reality of contention management, which
is an on going process. All our algorithms work in this general setting. 
For our lower bounds, we use a setting in which we restrict the environment
to only issue broadcast requests at the beginning of round one. We call 
this the {\em one-shot setting}.  \fullOnly{Note that this restriction only 
strengthens the lower bounds and it furthermore simplifies the notation.} Also, in most of our lower bounds, we consider, $G$ and $G'$ to be bipartite graphs, where nodes of one part are called \emph{senders} and they receive broadcast inputs, and nodes of the other part are called \emph{receivers}, and each have a sender neighbor. In this setting, when referring to contention $c(u)$, we furthermore mean $c(u,1)$. Note that in this setting, for any $r, r'$, $c(u,[r,r'])$ is less than or equal to $c(u,1)$. The same holds for $c'(u)$. Also, in these bipartite networks, the maximum $G'$-degree (or $G$-degree in the classical model) of the receiver nodes provides an upper bound on the maximum contention $\Delta'$ (or $\Delta$ in the classical model). When talking about these networks, and when it is clear from the context, we sometimes use the phrase {\em maximum receiver degree} instead of the maximum contention.

%% file: RelatedWork.tex
\section{Related Work}

\fullOnly{{\bf Single-Hop Networks}: The $k$-selection problem is the restricted case of the local broadcast problem for single-hop networks, in classical model. This problem is defined as follows. The network is a clique of size $n$, and $k$ arbitrary processes are active with messages. The problem is for all of these active processes to deliver their messages to all the nodes in the network. This problem received a vast range of attention throughout 70's and 80's, and under different names, see \emph{e.g.}  \cite{TM77}- \cite{GW85}. For this problem, Tsybakov and Mikhailov \cite{TM77}, Capetanakis \cite{C79a, C79b}, and Hayes \cite{H79}, (independently) presented deterministic tree algorithms with time complexity of $O(k + k \log (\frac{n}{k}))$ rounds. Komlos and Greenberg \cite{KG85} showed if processes know the value of $k$, there exists algorithms that work with the same time complexity in networks that do not provide any collision detection mechanism. Greenberg and Winograd \cite{GW85} showed a lower bound of $\Omega(\frac{k \log n}{\log k})$ for time complexity of deterministic solutions of this problem in the case of networks with collision detection. 

On the other hand, Tsybakov and Mikhailov \cite{TM77}, and Massey \cite{M80}, and Greenberg and Lander \cite{GL83} present randomized algorithms that solve this problem in expected time of $O(k)$ rounds. One can see that with simple modifications, these algorithms yield high-probability randomized algorithms that have time complexity of $O(k)+ polylog(n)$ rounds.\\}

\fullOnly{{\bf Multi-Hop Networks}:} Chlamatac and Kutten~\cite{CK85} were the first to introduce the classical radio network model. Bar-Yehuda et al. \cite{BGI87} studied the theoretical problem of local broadcast in synchronized multi-hop radio networks as a submodule for the broader goal of global broadcast. For this, they introduced {\em Decay} procedure, a randomized distributed procedure that solves the local broadcast problem. Since then, this procedure has been the standard method for resolving contention in wireless networks (see \emph{e.g.}~ \cite{GPX05, KLN09, KKKL10, KKLMP11}). In this paper, we prove that a slightly modified version of Decay protocol achieves optimal progress and acknowledgment bounds in both the classical radio network model and the dual graph model. A summary of these time bounds is presented in Figure~\ref{fig:results}.

Deterministic solutions to the local broadcast problem are typically based on combinatorial objects called \emph{Selective Families}, see \emph{e.g.} \cite{CGGPR00}-\cite{CMS04}. Clementi et al.  \cite{CMS01} construct $(n, k)$-selective families of size $O(k \log n)$ (\cite[Theorem 1.3]{CMS01}) and show that this bound is tight for these selective families (\cite[Theorem 1.4]{CMS01}). Using these selective families, one can get local broadcast algorithms that have progress bound of $O(\Delta \log n)$, in the classical model. These families do not provide any local broadcast algorithm in the dual graph model. Also, in the same paper, the authors construct $(n,k)$-strongly-selective families of size $O(k^2 \log n)$ (\cite[Theorem 1.5]{CMS01}). They also show (in \cite[Theorem 1.6]{CMS01}) that this bound is also, in principle, tight for selective families when $k \leq \sqrt{2n}-1$. Using these strongly selective families, one can get local broadcast algorithms with acknowledgment bound of $O(\Delta^2 \log n)$ in the classical model and also, with acknowledgment bound of $f_{ack}(k)=O((\Delta^\prime)^{2} \log n)$ in the dual graph model. As can be seen from our results (summarized in Figure~\ref{fig:results}), all three of the above time bounds are far from the optimal bounds of the local broadcast problem. This shows that when randomized solutions are admissible, solutions based on these notions of selective families are not optimal.

In \cite{CCMPS01}, Clementi et al. introduce a new type of selective families called Ad-Hoc Selective Families which provide new solutions for the local broadcast problem, if we assume that processes know the network. Clementi et al. show in \cite[Theorem 1]{CCMPS01} that for any given collection $\mathcal{F}$ of subsets of set $[n]$, each with size in range $[\Delta_{min}, \Delta_{max}]$, there exists an ad-hoc selective family of size $O((1+\log(\Delta_{max}/\Delta_{min})) \cdot \log |F|)$. This, under the assumption of processes knowing the network, translates to a deterministic local broadcast algorithm with progress bound of $O(\log\Delta \, \log n)$, in the classical model. This family do not yield any broadcast algorithms for the dual graph model. Also, in \cite{CMS04}, Clementi et al. show that for any given collection $\mathcal{F}$ of subsets of set $[n]$, each of size at most $\Delta$, there exists a Strongly-Selective version of Ad-Hoc Selective Families that has size $O(\Delta \log |F|)$ (without using the name ad hoc). This result shows that, again under the assumption of knowledge of the network, there exists a deterministic local broadcast algorithms with acknowledgment bounds of $O(\Delta \log n)$ and $O(\Delta' \log n)$, respectively in the classical and dual graph models. Our lower bounds for the classical model show that both of the above upper bounds on the size of these objects are tight.


%% file: UpperBounds.tex
\section{Upper Bounds for Both Classical and Dual Graph Models}\label{sec:upper}
In this section, we show that by slight modifications to Decay protocol, we can achieve upper bounds that match the lower bounds that we present in the next sections. \shortOnly{Due to space considerations, the details of the related algorithms are omitted from the conference version and can be found in~\cite{GHLN12}. 

\begin{theorem} In the classical model, there exists a distributed local broadcast algorithm that gives acknowledgment bound of $f_{ack}(k) = O(\Delta \log n)$ and progress bound of $f_{prog}(k) = O(\log \Delta \log n)$. \end{theorem}

\begin{theorem}There exists a distributed local broadcast algorithm that, in the classical model, gives bounds of $f_{ack}(k)=O(\Delta \log n)$ and $f_{prog}(k)= O(\log \Delta \log n)$, and in the dual graph model, gives bounds of $f_{ack}(k) = O(\Delta' \log n)$ and $f_{prog}(k) = O(\min\{k \log \Delta' \log n, \Delta' \log n\})$.\end{theorem}

\begin{theorem} In the dual graph model, there exists a distributed local broadcast algorithm that gives acknowledgment bound of $f_{ack}(k) = O(\Delta' \log n)$ and progress bound of $f_{prog}(k) = O(\min\{k \log k \log n, \Delta' \log n\})$. \end{theorem} 
}
\fullOnly{
We first present three local broadcast algorithms. The first algorithm, the Synchronous Acknowledgment Protocol (SAP), yields a good acknowledgment bound and the other two algorithms, Synchronous Progress Protocol (SPP) and Asynchronous Progress Protocol (APP), achieve good progress bounds. From these two progress protocols, the SPP protocol is exactly the same as the \emph{Decay procedure} in \cite{BGI87}. In that paper, this protocol was designed as a submodule for the global broadcast problem in the classical model. Here, we reanalyze that protocol for the Dual Graph model. Furthermore, the APP protocol is similar to the Harmonic Broadcast Algorithm in \cite{KLNOR10}. In that work, the Harmonic Broadcast Algorithm is introduced and used as a solution to the problem of global broadcast in the dual graph model. We analyze the modified version of this algorithm, which we call the APP protocol, and show that it yields good progress bounds in the dual graph model. Then, we show how to combine the acknowledgment and progress protocols to get both fast acknowledgment and fast progress. Particularly, one can look at the combination of SAP and SPP as an optimized version of the Decay procedure, adjusted to provide tight progress and acknowledgment together.

\subsection{The Synchronous Acknowledgment Protocol (SAP)}\label{subsec:SAP}
In this section, we present the SAP protocol and show that this algorithm has acknowledgment bounds of $O(\Delta' \log n)$ and $O(\Delta \log n)$, respectively, in the dual graph and the classical model. \fullOnly{The reason that we call this algorithm synchronous is that the rounds are divided into contiguous sets named epochs, and the epochs of different processes are synchronized (aligned) with each other.} In the SAP algorithm, each epoch consists of $\Theta(\Delta' \log n)$ rounds. Whenever a process receives a message for transmission, it waits till the start of next epoch. If a process $v$ has received input $bcast(m)_v$ before the start of an epoch and has not output $ack(m)_v$ by that time, we say that in that epoch, process $v$ is \emph{ready with message $m$} or simply \emph{ready}. 

As presented in Algorithm \ref{alg:SAP}, each epoch of SAP consists of $\log \Delta'$ phases as follows. For each $i \in [\log \Delta']$, the $i^{th}$ phase is comprised of $\Theta(2^i \log n)$ rounds where in each such round, each ready process transmits with probability $\frac{1}{2^i}$. After the end of the epoch, each ready process acknowledges its message.

\begin{algorithm}[th]\label{alg:SAP}
\caption{An epoch of SAP in process $v$ when $v$ is ready with message $m$}
\begin{algorithmic}
\footnotesize
\For {$i$ := $1$ to $\log \Delta'$} 
	\For {$j$:= $1$ to $\Theta(2^i \log n)$}
		\State transmit $m$ with probability $\frac{1}{2^i}$
	\EndFor
\EndFor
\State output $ack(m)_v$
\end{algorithmic}
\end{algorithm}

\begin{lemma} \label{lemma:SAP-Dual} The Synchronous Acknowledgment Protocol solves the local broadcast problem in the dual graph model and has acknowledgment time of $O(\Delta' \log n)$.
\end{lemma}
\begin{proof} Consider a process $v$ a round $r$ such that $v$ receives an input of $bcast(m)_v$ in round $r$. First, note that process $v$ acknowledges message $m$ by at most two epochs after round $r$, \emph{i.e.}, process $v$ outputs an $ack(m)_v$ by time $r' = r + \Theta(\Delta' \log n)$. Now assume that epoch $\wp$ is the epoch that $v$ becomes ready with $m$. In order to show that SAP solves the local broadcast problem, we claim that by the end of epoch $\wp$, with high probability, $m$ is delivered to all the processes in $\mathcal{N}_{G}(v)$. Consider an arbitrary process $u \in \mathcal{N}_{G}(v)$. To prove this claim, we show that by the end of epoch $\wp$, with high probability, $u$ receives $m$. A Union Bound then completes the proof of the claim.

To show that $u$ receives $m$ by the end of epoch $\wp$, we focus on the processes in $\mathcal{N}^{+}_{G\,'}(u)$. Suppose that $r''$ is the last round of epoch $\wp$ and that the number of ready processes in $\mathcal{N}^{+}_{G'}(u)$ during this epoch is at most $k=c(u, r, r'')$. Now, consider the phase $i = \lfloor \log k \rfloor$ of epoch $\wp$. In each round of this phase, the probability that $u$ receives the message of $v$ is at least $ \frac{1}{2^i} \;(1-\frac{1}{2^i})^{k} \approx \frac{1}{k}\; e^{-\,\frac{k}{k}} = \frac{1}{e\cdot k}$, where the first term of the LHS is the probability of transmission of process $v$ and the other term is the probability that rest of the ready processes in $\mathcal{N}^{+}_{G'}(u)$ remain silent. Now, phase $i$ has $\Theta(2^i \log n) = \Theta(k \log n)$ rounds. Therefore, the probability that $u$ does not receive $m$ in phase $i$ is at most $(1-\frac{1}{e\cdot k})^{\Theta(k \log n)} = e^{-\Theta(\log n)} = (\frac{1}{n})^{\Theta(1)}$
Hence, the probability that $u$ does not receive the message $m$ in epoch $\wp$ is $(\frac{1}{n})^{\Theta(1)}$. This completes the proof. 
\end{proof}

\begin{corollary} \label{crl:SAP-classical}The SAP protocol solves the local broadcast problem in the classical model and has an acknowledgment bound of $O(\Delta \log n)$.
\end{corollary}
\begin{proof} The corollary can be easily inferred from Lemma \ref{lemma:SAP-Dual} by setting $G=G'$.
\end{proof}
\subsection{The Synchronous Progress Protocol (SPP)}\label{subsec:SPP}
In this section, we present and analyze the SPP protocol, which is also known as decay procedure. From Theorem 1 in \cite{BGI87}, it can be inferred that this protocol achieves a progress bound of $O(\log\Delta \, \log n)$ in the classical model. Here, we reanalyze this protocol with a specific focus on its progress bound in the dual graph model. More specifically, we show that this protocol yields a progress bound of $f_{prog}(k) = O(k \log(\Delta') \log n)$ in the dual graph model. 
%
%

Similar to the SAP protocol, the rounds of SPP are divided into contiguous sets called epochs\fullOnly{ and epochs of different processes are synchronized with each other}. The length of each epoch of SPP is $\log \Delta'$ rounds. Similar to the SAP protocol, whenever a process $v$ receives a message $m$ for transmission, by getting input $bcast(m)_v$, it waits till the start of next epoch. Moreover, if input $bcast(m)_v$ happens before the start of an epoch and process $v$ has not outputted $ack(m)_v$ by that time, we say that in that epoch, process $v$ is \emph{ready with message $m$} or simply \emph{ready}. 

As presented in Algorithm \ref{alg:SPP}, in each epoch of SPP and for each round $i \in [\log \Delta']$ of that epoch, each ready process transmits its message with probability $\frac{1}{2^i}$. Each process acknowledges its message $\Theta(\Delta' \log n)$ epochs after it receives the message.

\begin{algorithm}[th]\label{alg:SPP}
\caption{The procedure of SPP in process $v$ when $v$ becomes ready with message $m$}
\begin{algorithmic}
\footnotesize
\For {$j$ := $1$ to $\Theta(\Delta' \log n)$} 
\For {$i$ := $1$ to $\log \Delta'$} \hfill {/*Each turn of this loop is one epoch*/}
\State transmit $m$ with probability $\frac{1}{2^i}$
\EndFor
\EndFor
\State output $ack(m)_v$
\end{algorithmic}
\end{algorithm}

\fullOnly{From the above description, it is clear that the general approach used in the protocols SAP and SPP are similar. In both of these protocols, in each round $r$, each ready process transmits with some probability $p(r)$ and this probability only depends on the protocol and the round number, i.e., the probabilities of transmissions in different ready processes are equal. Also, one can see that in round $r$, a node $u$ has the maximum probability of receiving some message if $c(u, r)$ is around $\frac{1}{p(r)}$. Hence, having rounds with different transmission probabilities is like aiming for nodes that have different levels of contention, i.e., $c(u, r)$. Noting this point, we see that the core difference between the SAP and SPP protocols is as follows. In the SAP, each epoch starts with a phase of rounds all aimed at nodes with smaller contention. The number of rounds in this phase is designed so that all the nodes at that contention level receive all the messages that are under transmission in their $G$-neighborhood. Then, after clearing out one level of contention, SAP goes to the next level, and it continues this procedure till cleaning up the nodes at largest level of contention. On the other hand, SPP is designed so that makes progress on all levels of contention gradually and altogether. That is, in each epoch of SPP, which is much shorter than those in APP, all the levels of contention are aimed for exactly once. 

Now, we show that because of this property, SPP has a good progress bound.

}

\begin{lemma} \label{lem:SPP} The synchronous progress protocol solves the local broadcast problem in the dual graph model and provides progress bound of $f_{prog}(k)=O(k \log (\Delta') \log n)$. \fullOnly{Also, SPP provides receive bound of $f_{rcv}(k)=O(k \log (\Delta')\, \log n)$.}
\end{lemma}
\fullOnly{\begin{proof}
It is clear that in SPP, each message is acknowledged after $\Theta(\Delta' \log n)$ epochs and therefore after $O(\Delta' \log(\Delta') \log n)$ rounds. Similar to the proof of Lemma \ref{lemma:SAP-Dual}, we can easily see that each acknowledged message is delivered to all the $G$-neighbors of its sender. Thus, SPP solves the local broadcast problem.

Now, we first show that SPP has progress time of $f_{prog}(k) = \Theta(k \log (\Delta') \log n)$. Actually, we show something stronger. We show that within the same time bound, $u$ receives the messages of each of its ready $G$-neighbors. For this, suppose that there exists a process $u$ and a round $r$ such that in round $r$, at least one process $w \in \mathcal{N}_{G}(u)$ has a message for transmission such that $u$ has not received it. Also, suppose that the first round after $r$ that $u$ receives the message of $w$ is round $r'$. Such a round exists w.h.p as the SPP solves the broadcast problem. Let $k= c(u, r, r')$, i.e., the total number of processes in $\mathcal{N}_{G\,'}(u)$ that are ready in at least one round in range $[r, r']$. We show that $r' \leq r + \Theta(k \log (\Delta') \log n)$. 

Suppose that $P$ consists of all the epochs starting with the first epoch after round $r$ and ending with the epoch that includes round $r'$. If $P$ has less than $\Theta(k \log \log n)$ epochs, we are done with the proof. On the other hand, assume that $P$ has at least $\Theta(k \log \log n)$ epochs. Let $i = \lfloor \log k \rfloor$. Now, for the $i^{th}$ round of each epoch in $P$, the probability that $u$ receives the message of $w$ in that round is at least 
\[ \frac{1}{2^i} \;(1-\frac{1}{2^i})^{k} \approx \frac{1}{k}\; e^{-\,\frac{k}{k}} = \frac{1}{e\cdot k}\]
Therefore, the probability that $u$ does not receive $w$'s message in the $\Theta(k \log n)$ epochs of $P$ is at most 
\[ (1-\frac{1}{e\cdot k})^{\Theta(k \log n)} = e^{-\Theta(\log n)} = (\frac{1}{n})^{\Theta(1)}\]

To see the second part of the lemma, suppose that process $v$ receives input $bcast(m')_v$ in round $\tau$ and outputs $ack(m')_v$ in round $\tau'$. Let $k' = c'(v, \tau, \tau')$. We argue that all processes in $\mathcal{N}_{G}(v)$ receive $m'$ by round $ \tau''= \tau + \Theta(k' \log (\Delta') \log n)$. Using the above argument, we see that each process $u \in \mathcal{N}_{G}(v)$ receives the message of $v$ in time $O(c(u, r, r') \log (\Delta') \log n)$ where $r$ and $r'$ are defined as above for $u$ and also, we have $r,r' \in [\tau, \tau']$. Moreover, by definition of the $c'(v, \tau, \tau')$, for each $u \in \mathcal{N}_{G}(v)$, we have $c(u, r, r') \leq c'(v, \tau, \tau') = k'$. Thus, all neighbors of $v$ receive $m'$ by time $\tau''$. This completes the proof of the second part.
\end{proof}}

\fullOnly{\begin{lemma} \label{lemma:SPP-classical} The SPP protocol solves the local broadcast problem in the classical model and gives a progress bound of $O(\log(\Delta) \log n)$. 
\end{lemma} 
\begin{proof} This bound can also be inferred from Theorem 1 in \cite{BGI87}. For the sake of completeness, and since analysis are simple and similar to the previous ones, we present the complete version here. 

Similar to Corollary \ref{crl:SAP-classical}, we can easily see that the SPP protocol solves the local broadcast problem in the classical model from the result about the dual graph model by setting by setting $G'=G$ in the Lemma \ref{lem:SPP}. To see the progress time bound, consider process $u$ and suppose that there is a round $r$ in which some process in $\mathcal{N}_G(u)$ has a message for transmission such that process $u$ has not received it so far. Also, let $r'$ be the first round after $r$ that $u$ receives a message. Again, such a round exists since SPP solves the local broadcast problem. Let $k = c(u, r, r')$ and $i = \lfloor \log k \rfloor$. The probability that $u$ receives a new message in $i^{th}$ round of the each epoch after round $r$ is at least $\frac{c(u)}{2^i} \;(1-\frac{1}{2^i})^{k} \approx \; e^{-\,\frac{k}{k}} = \frac{1}{e}$. Therefore, the probability that $r' > r+ \Theta(\log n \, \log(\Delta))$ is at most $(1-\frac{1}{e})^{\Theta(\log n)} = (\frac{1}{n})^{\Theta(1)}$. This completes the proof.
\end{proof}}
\shortOnly{

{\noindent\bf The Asynchronous Progress Protocol (APP):}}\fullOnly{\subsection{The Asynchronous Progress Protocol (APP)}\label{subsec:APP}}
In this section, we present and study the APP protocol and show that it yields progress bound of $f_{prog}(k) = O(k \log(k) \log n)$ in dual graph model. Note that this is better than the bound achieved in SPP. However, in comparison to the bound achieved by SPP in the classical model, APP does not guarantee a good progress time. This protocol is, in principle, similar to the Harmonic Broadcast Algorithm in \cite{KLNOR10} that is used for global broadcast in the dual graph model.

Similar to the SAP and SPP protocols, the rounds of APP are divided into epochs as well. However, in contrast to those two protocols, and as can be inferred from the name, the epochs of APP in different processes are not synchronized with each other. Also, in APP, a process $v$ becomes \emph{ready} immediately after it receives the $bcast(m)_v$ input. 
\fullOnly{

}
Whenever a process becomes ready, it starts an epoch as follows. This epoch consists of $\log \Delta' + \log \log \Delta'$ phases. For each $i \in [\log \Delta'+ \log \log \Delta']$, the $i^{th}$ phase is comprised of $\Theta(2^i \log n)$ rounds where in each such round, each ready process transmits with probability $\frac{1}{2^i}$. Also, the process outputs $ack(m)_v$ at the end of this epoch.

\begin{algorithm}[th]\label{alg:APP}
\caption{An epoch of APP in process $v$ when $v$ is ready with message $m$}
\begin{algorithmic}
\For {$i$ := $1$ to $\log \Delta' + \log \log \Delta'$ }
	\For {$j$:= $1$ to $\Theta(2^i \log n)$} 
		\State transmit $m$ with probability $\frac{1}{2^i}$
	\EndFor
\EndFor
\State output $ack(m)_v$
\end{algorithmic}
\end{algorithm}

\begin{lemma} \label{lem:APP} The asynchronous progress protocol solves the local broadcast problem in the dual graph model and has progress time of $f_{prog}(k)= O(k \log (k) \log n)$. \fullOnly{Also, APP achieves receive bound of $f_{rcv}(k)=O(k \log (k) \log n)$.}
\end{lemma}
\fullOnly{\begin{proof} Suppose that there exists a process $u$ and a round $r$ such that in round $r$, some process $v$ in $\mathcal{N}_{G}(u)$ has a message $m$ that is not received by $u$, i.e., $m$ is new to $u$. Let $r'$ be an arbitrary round after round $r$ and let $R$ be the set of all rounds in range $[r, r']$. So, we have $r' = r+|R|-1$. Then, let $k= c(u, r, r')$. In order to prove the progress bound part of the theorem, we show that if $r' - r\geq \Theta(k \cdot\log k \cdot\log n)$, then, with high probability, $u$ receives $m$ by round $r'$. Note that this is even stronger than proving the claimed progress bound because this means that $u$ receives each of the new messages (new at round $r$) by $r'$. 
Since $k$ can be at most $\Delta'$, this would automatically show that APP solves the local broadcast problem. Also, similar to the proof of Lemma \ref{lem:SPP}, this would prove the second part of the theorem as well.


Let $S$ be the set of all processes in $\mathcal{N}_{G'}(u) - \{v\}$ that are ready in at least one round of $R$. Therefore, we have $|S|=k-1$. First, since the acknowledgment in APP takes a full epoch, during rounds of $R$, each process in $S$ transmits at most a constant number of messages and therefore, the total number of messages under transmission in $\mathcal{N}_{G'}(u)$ in rounds of $R$ is $O(k)$. 

We show that w.h.p. $u$ receives message $m$ by the end of rounds of $R$. In order to do this, we divide the rounds of $R$ into two categories of \emph{free} and \emph{busy}. Similar to \cite{KLNOR10}, we call a round $\tau$ \emph{busy} if the total probability of transmission of processes of $S$ in round $\tau$ is greater than or equal to $1$. Otherwise, the round $\tau$ is called \emph{free}. Similar to \cite[Lemma 11]{KLNOR10}, we can see that the total number of busy rounds in set $R$ is $O(k \cdot\log k \cdot\log n)$. Therefore, there are $\Theta(k \cdot \log k \cdot\log n)$ free rounds in $R$. On the other hand, similar to \cite[Lemma 11]{KLNOR10}, we can easily see that if $\tau \in R$ is a free round and the probability of transmission of $v$ in round $r$ is $p_v(\tau)$, then $u$ receives the message of process $u$ in round $\tau$ with probability at least $\frac{1}{4 p_v(\tau)}$. Now, because of the way that SPP chooses its probabilities and since $|R| = \Theta(k \;\log k \;\log n)$, we can infer that the transmission probability of $v$ for each round $\tau \in R$ is at least $\frac{1}{2k \,\log k}$. Therefore, since $R$ has $\Theta(k \cdot \log k \cdot\log n)$ free rounds and for each free round $\tau \in R$, $u$ receives the message of $v$ with probability at least $\frac{1}{4 p_v(\tau)}$, we can conclude that the probability that $u$ does not receive the message of $v$ by the end of rounds of $R$ is at most
\[ (1-\frac{1}{4 k \,\log k})^{\Theta(k \; \log k \;\log n)} \leq e^{-\Theta(\log n)} = (\frac{1}{n})^{\Theta(1)}\]
This completes the proof. 
\end{proof}
}
\shortOnly{

{\noindent\bf {Interleaving Progress and Acknowledgment Protocols:}}}\fullOnly{
\subsection{Interleaving Progress and Acknowledgment Protocols}}
\fullOnly{In this section, we show how we can achieve both fast progress and fast acknowledgment bounds by combining our acknowledgment protocol, SAP, with either of the progress protocols, SPP or APP.} The general outline for combining the above algorithms is as follows. Suppose that we want to combine the protocol SAP with a protocol $P_{prog} \in \{SPP, APP\}$. Then, whenever process $v$ receives message $m$ for transmission, by a $bsast(m)_v$ input event, we provide this message as input to both of the protocols SAP and $P_{prog}$. Then, we run the SAP protocol in the odd rounds, and protocol $P_{prog}$ in the even rounds. In the combined algorithm, process $v$ acknowledges the message $m$ by outputting $ack(m)_v$ in the round that SAP acknowledges $m$. Moreover, in that round, the protocol $P_{prog}$ also finishes working on this message. \fullOnly{In the following, we show that using}\shortOnly{Using} this combination, we achieve the fast progress and acknowledgment bounds together. \fullOnly{More formally, we show that  the acknowledgment and progress time of the combined algorithm are respectively, two times the minimums of the acknowledgment and two times the minimum of the progress times of the respective two protocols.} 

\fullOnly{\begin{lemma} \label{lem:Comb-Ack}If we interleave the SAP protocol with a protocol $P_{prog} \in \{SPP, APP\}$, the resulting algorithm solves the local broadcast problem and has acknowledgment bound of $f_{ack}(k) = O(\Delta' \log n)$ in the dual graph model, and acknowledgment bound of $f_{ack}(k) = O(\Delta \log n)$ in the classical model.
\end{lemma}}

\fullOnly{\begin{proof}First, note that the even and odd rounds of different processes are aligned and therefore, in each round, only one of the protocols SAP and $P_{prog}$ is transmitting throughout the whole network. Because of this, it is clear that when in process $v$, the SAP protocol acknowledges message $m$, $m$ is successfully delivered to all the processes in $\mathcal{N}_{G}(v)$. Now, suppose that process $v$ receives an input $bcast(m)_v$ in round $r$. Using Lemma \ref{lemma:SAP-Dual}, we know that the SAP protocol acknowledges message $m$ by $\Theta(\Delta' \log n)$ odd rounds after $r$. Thus, process $v$ outputs $ack(m)_v$ in a round $r'=r+\Theta(\Delta' \log n)$. Hence, we have that the interleaved algorithm solves the local broadcast problem and has acknowledgment bounds of $f_{ack}(k) = \Theta(\Delta' \log n)$ and $f_{ack}(k) = \Theta(\Delta \log n)$, respectively for, the dual graph and the classical radio broadcast models.
\end{proof}
}
\begin{corollary}If we interleave SAP with SPP, in the dual graph model, we get acknowledgment bound of $f_{ack}(k) = O(\Delta' \log n)$ and progress bound of $f_{prog}(k) = O(\min\{k \log (\Delta') \log n, \Delta' \log n\})$. Also, this interleaving gives acknowledgment and progress bounds of, respectively, $O(\Delta \log n)$ and $O(\log (\Delta) \log n)$ in the classical radio broadcast model. \end{corollary}

\fullOnly{\begin{proof} The acknowledgment bound parts of the corollary follow immediately from Lemma \ref{lem:Comb-Ack}. For the progress bound parts, consider a process $u$ and a round $r$ such that there exists a process $v \in \mathcal{N}_{G}(u)$ that is transmitting message $m$ and process $u$ has not received message $m$ before round $r$. Note that for each $r'>r$, if we have $c(u, r, r')=k$, then, by definition of $c(u, r, r')$, in each even round $\tau \in [r, r']$, we have $c(u, \tau) \leq k$. The rest of the proof follows easily from Lemmas \ref{lem:SPP} and \ref{lemma:SPP-classical}, and by focusing on the SPP protocol in the even rounds after $r$.
\end{proof}
} 
\begin{corollary}If we interleave SAP with APP, in the dual graph model, we get acknowledgment bound of $f_{ack}(k) = O(\Delta' \log n)$ and progress bound of $f_{prog}(k) = O(\min\{k \log (k) \log n, \Delta' \log n\})$. \end{corollary} 

\fullOnly{

\begin{proof} Again, the acknowledgment bound part of the corollary follows immediately from Lemma \ref{lem:Comb-Ack}. For the progress part, consider a process $u$ and a round $r$ such that there exists some process $v \in \mathcal{N}_{G}(u)$ that is transmitting a message $m$ and process $u$ has not received message $m$. Suppose that $r'$ is the first round that $u$ receives $m$. Such a round exists with high probability as we know from Lemma \ref{lem:Comb-Ack} that the combined algorithm solves the local broadcast problem. Let $k = c(u, r, r')$.  Let $r'' = r + \Theta(\min\{k \log (k) \log n, \Delta' \log n\})$. If $r' \leq r''$, we are done with the proof. In the more interesting case, suppose that $r' < r''$.

Now, by definition of $c(u, r, r')$, we know that in each even round $\tau \in [r, r']$, we have $c(u, \tau) \leq k$. Hence, we also have that in each even round $\tau \in [r, r'']$, $c(u, \tau) \leq k$. Let $S$ be the set of processes in $\mathcal{N}_{G\,'}(u)$ that are active in at least one even round in range $[r, r'']$. Thus, $|S|\leq k$. Since $r'' - r \leq \Theta(\Delta' \log n)$ and the algorithm acknowledges each message after $\Theta( \Delta' \log n)$ rounds, during even rounds in range $[r, r'']$, each process $w \in S$ transmits only a constant number of messages. Therefore, the total number of messages under transmission during even rounds in range $[r, r'']$ is $O(k)$. The rest of the proof can be completed exactly as that in the proof of Lemma \ref{lem:APP}. 
\end{proof}
}

}

%% file: Classic.tex
\section{Lower Bounds in the Classical Radio Broadcast Model}
In this section, we focus on the problem of local broadcast in the classical model and present lower bounds for both progress and acknowledgment times. We emphasize that all these lower bounds are presented for centralized algorithms and also, in the model where processes are provided with a collision detection mechanism. Note, that these points only strengthen these results. These lower bounds prove that the optimized decay protocol, as presented in the previous section, is optimal with respect to progress and acknowledgment times in the classical model. These lower bounds also show that the existing constructions of Ad Hoc Selective Families are optimal. Moreover, in future sections, we use the lower bound on the acknowledgment time in the classical model that we present here as a basis to derive lower bounds for progress and acknowledgment times in the dual graph model. 
\subsection{Progress Time Lower Bound}
In this section, we remark that following the proof of the $\Omega(\log^2 n)$ lower bound of Alon et al. \cite{ABLP89} on the time needed for global broadcast of one message in radio networks, and with slight modifications, one can get a lower bound of $\Omega(\log\Delta \log n)$ on the progress bound in the classical model.

\begin{lemma}
For any $n$ and any $\Delta \leq n$, there exists a one-shot setting with a bipartite network of size $n$ and maximum contention of at most $\Delta$ such that for any transmission schedule, it takes at least $\Omega(\log\Delta \log n)$ rounds till each receiver receives at least one message.
\end{lemma}\fullOnly{
\begin{proof}[Proof Outline]
The proof is an easy extension of \cite{ABLP91} to networks with maximum contention of $\Delta$. The only change is that instead of choosing the receiver degrees to vary between $n^{0.4}$ and $n^{0.6}$, we choose the degrees between $14$ and $\Theta(\sqrt{\Delta})$. This leads to $\log(\Delta)$ (instead of $\log n$) different classes of degrees, and in turn, to the stated bound. The proof stays mostly unaffected.
\end{proof}}
\subsection{Acknowledgment Time Lower Bound}
In this section, we present our lower bound on the acknowledgment time in the classical radio broadcast model.

\begin{theorem} \label{thm:Ack_LB}
In the classical radio broadcast model, for any large enough $n$ and any $\Delta \in [20 \log n,  n^{0.1}]$, there exists a one-shot setting with a bipartite network 
of size $n$ and maximum receiver degree at most $\Delta$ such that it takes at least $ \Omega(\Delta \log n)$ rounds until all receivers have received all messages of their sender neighbors.
\end{theorem} 

To prove this theorem, instead of showing that randomized algorithms have low success probability, we show a stronger variant by proving an impossibility result: we prove that there exists a one-shot setting with the above properties such that, even with a centralized algorithm, it is \emph{not possible} to schedule transmissions of nodes less than some bound of $\Omega(\Delta \log n)$ rounds such that each receiver receives the message of each of its neighboring senders successfully. In particular, this result shows that in this one-shot setting, for any randomized local broadcast algorithm, the probability that an execution shorter than that $\Omega(\Delta \log n)$ bound successfully delivers message of each sender to all of its receiver neighbors is zero.  

Let us first present some definitions. A transmission schedule $\sigma$ of length $L(\sigma)$ for a bipartite network is a sequence $\sigma_1, \ldots, \sigma_{L(\sigma)} \subseteq S$ of senders. Having a sender $u \in \sigma_r$ indicates that at round $r$ the sender $u$ is transmitting its message. For a network $G$, we say that transmission schedule $\sigma$ \emph{covers} $G$ if for every $v \in S$ and $u \in \mathcal{N}_{G}(v)$, there exists a round $r$ such that $\sigma_r \cap \mathcal{N}_{G}(v) = \{u\}$, that is, using transmission schedule $\sigma$ every receiver node receives all the messages of all of its sender neighbors. Now we are ready to see the main lemma which proves our bound. 

\begin{lemma}\label{lem:Ack_Schedules_LB}
For any large enough $n$ and $\Delta \in [20 \log n, n^{0.1}]$, there exists a bipartite network $G$ with size $n$ and maximum receiver degree at most $\Delta$ such that there does not exist a transmission schedule $\sigma$ such that $L(\sigma) < \frac{\Delta\log n}{100}$ and $\sigma$ covers $G$.
\end{lemma}

The rest of this subsection is devoted to proving this lemma. As in the previous subsection, our proof uses techniques similar to those of \cite{ABLP91, ABLP89, ABLP92} and utilizes the probabilistic method~\cite{AS00} to show the existence of the network $G$ mentioned in the \Cref{lem:Ack_Schedules_LB}.

First, we fix an arbitrary $n$ and a $\Delta \in [20 \log n, n^{0.1}]$ and let $\eta= n^{0.12}$ and $m = \eta^8 = n^{0.96}$. Next, we present a probability distribution over a particular family $\mathcal{G}$ of bipartite networks. The common structure of this graph family $\mathcal{G}$ is as follows. All networks of $\mathcal{G}$ have a fixed set of nodes $V$. Moreover, $V$ is partitioned into two nonempty disjoint sets $S$ and $R$, which are respectively the set of senders and the set of receivers. We have $|S|=\eta$ and $|R|=m$. The total number of nodes in these two sets is $\eta+m = n^{0.12}+n^{0.96}$. We adjust the number of nodes to exactly $n$ by adding enough isolated senders to the graph. Instead of defining the probability mass distribution of these graphs we describe the process that samples networks from $\mathcal{G}$. A random sample network is simply created by independently putting an edge between any $s \in S$ and $r \in R$ with probability $\frac{\Delta}{2\eta}$. Given a random network from this distribution we first show that with high probability the maximum receiver degree is at most $\Delta$, as desired. 

\begin{lemma} \label{lem:degrees}For a random sample graph $G \in \mathcal{G}$, with probability at least $1-\frac{1}{n^2}$, the degree of any receiver node $r \in R$ is at most $\Delta$.
\end{lemma}
\begin{proof} For each $r \in R$, let $X_G(r)$ denote the degree of node $r$ in random sample graph $G$. Then, $\mathbb{E}[X_G(r)] = \eta \cdot \frac{\Delta}{2\eta} = \frac{\Delta}{2}$. Moreover, since edges are added independently, we can use a Chernoff bound and obtain that $\Pr[X_G(r) \geq \Delta] \leq e^{-\frac{\Delta}{6}}$. Using a union bound over all choices of receiver node $r$, and noting that $\Delta \geq 20\log n$, we get that 
\begin{eqnarray}\Pr [\exists r \in R \; s.t.\; X_G(r) \geq \Delta] \leq&& \eta^{8} \cdot e^{-\frac{\Delta}{6}} = e^{8\log \eta -\frac{\Delta}{6}} = e^{0.96 \log n -\frac{\Delta}{6}} \nonumber \\
<&& e^{0.96 \log n -3 \log n} \leq e^{-2\log n} =\frac{1}{n^2}\nonumber
\end{eqnarray}
\end{proof}

\medskip
Now, we study the behavior of transmission schedules over random graphs drawn from $\mathcal{G}$. For each transmission schedule $\sigma$, call $\sigma$ \emph{short} if $L(\sigma) < \frac{\Delta\log n}{100}$. Moreover, for any fixed short transmission schedule $\sigma$, let $P(\sigma)$ be the probability that $\sigma$ covers a random graph $G \in \mathcal{G}$. Using a union bound, we can infer that for a random graph $G \in \mathcal{G}$, the probability that there exists a short transmission schedule $\sigma$ that covers $G$ is at most sum of the $P(\sigma)$-s, when $\sigma$ ranges over all the short transmission schedules. Let us call this probability \emph{the total coverage probability}. In order to prove the lower bound, we show \Cref{lem:probs} about \emph{the total coverage probability}.  Note, that given Lemmas \ref{lem:degrees} and \ref{lem:probs}, using the probabilistic method~\cite{AS00}, we can infer that there exists a network $G \in \mathcal{G}$ such that $G$ has maximum receiver degree of at most $\Delta$ and no short transmission schedule covers $G$. This completes the proof of \Cref{lem:Ack_Schedules_LB}.

\begin{lemma} \label{lem:probs} $\sum_{\sigma \; s.t. L(\sigma) <\frac{\Delta\log n}{100}} P(\sigma) \leq e^{-\sqrt{n}} \ll e^{-2\log n}=\frac{1}{n^2}$.
\end{lemma}

\begin{proof}
Note, that the total number of distinct short transmission schedules is less than $2^{\eta^3}$. This is because in each round there are $2^\eta$ options for selecting which subset of senders transmits. Then, each short transmission schedule has at most $\frac{\Delta\log n}{100} < \eta^2$ rounds. Therefore, the total number of ways in which one can choose a short transmission schedule is less than $2^{\eta^3}$. In order to prove that the total coverage probability is $e^{-\sqrt{n}}$, since the total number of short transmission schedules is less than $2^{\eta^3}=2^{n^{0.36}}$, it is enough to show that for each short transmission schedule $\sigma$, $P(\sigma) \leq e^{-n^{0.72}}$ as then the summation would be at most $2^{n^{0.36}} \cdot e^{-n^{0.72}} \leq e^{n^{0.36}-n^{0.72}} < e^{n^{-0.5}} = e^{-\sqrt{n}}$. Thus, it remains to prove that for each short transmission schedule $\sigma$, $P(\sigma)\leq e^{-n^{0.72}}$.

Fix an arbitrary short transmission schedule $\sigma$. for each round $t$ of $\sigma$, let $N(t)$ denote the number of senders that transmit in round $t$. Also, call round $t$ \emph{isolator} if $N(t)=1$. For each sender $s\in S$, if there exists an isolator round in $\sigma$ where only $s$ transmits in that round, then call sender $s$ \emph{lost}. Since $L(\sigma) \leq \frac{\Delta \log n}{100} \leq \frac{n^{0.1} \log n}{100} < \frac{n^{0.12}}{2} = \frac{\eta}{2}$, there are at least $\frac{\eta}{2}$ senders that \emph{are not lost}. 

For each not-lost sender $s$, we define a potential function $\Phi(s) = \sum_{t\in T_s} \frac{1}{N(t)}$  where $T_s$ is the set of rounds in which $s$ transmits. Note, that for each round $t$, the total potential given to not-lost senders in that round is at most $N(t) \cdot \frac{1}{N(t)} = 1$. Hence, the total potential when summer-up over all rounds is at most $\frac{\Delta \log n}{100} = \frac{\Delta \log \eta}{12}$. Therefore, since there are at least $\frac{\eta}{2}$ not-lost senders, there exists a not-lost sender $s^*$ for which $\Phi(s^*) \leq \frac{\Delta \log \eta}{6\eta }$.



Now we focus on sender $s^*$ and rounds $T_{s^*}$. We show that, for each receiver $r \in R$, there is a probability at least $\frac{1}{\eta^2}$ that node $r$ is a neighbor of $s$ and it does not receive message of $s^{*}$. First note that the probability that $r$ is a neighbor of $s$ is $\frac{\Delta}{\eta} > \frac{1}{\eta}$. Now for each $t \in T_{S^*}$, the probability that $r$ is connected to a sender other than $s^*$ that transmits in round $t$ is $1 - (1-\frac{\Delta}{2\eta})^{N(t)-1} \geq 1 - e^{-\frac{\Delta}{2\eta} \cdot (N(t)-1)} \geq 1 - e^{-\frac{\Delta}{4\eta} \cdot N(t)} \geq e^{-\frac{4\eta}{\Delta} \cdot \frac{1}{N(t)}}$. Thus, by the FKG inequality\cite[Chapter 6]{AS00}, the probability that this happens for every round $t \in T_{s^{*}}$ is at least $e^{-\sum_{t\in T_{s^*}} \frac{4\eta}{\Delta} \cdot \frac{1}{N(t)}} = e^{-\frac{4\eta}{\Delta} \cdot \Phi(s^*)}$. By choice of $s^*$, we know that this probability is greater than $e^{-\log \eta} = \frac{1}{\eta}$.
Hence, for each receiver $r$, the probability that $r$ is a neighbor of $s^{*}$ but never receives a message from $s$ is greater than $\frac{1}{\eta} \cdot \frac{1}{\eta} = \frac{1}{\eta^2}$. Given this, since edges of different receivers are chosen independently, the probability that there does not exist a receiver $r$ which satisfies above conditions is at most $(1-\frac{1}{\eta^2})^{\eta^8} \geq e^{-\eta^6}$. This shows that $P(\sigma)\leq e^{-\eta^6} = e^{-n^{0.72}}$ and thus completes the proof.

\end{proof}

%% file: Dual.tex
\section{Lower Bounds in the Dual Graph Model}
In this section,\fullOnly{ we present two lower bounds for the dual graph model. We}\shortOnly{ we} show a lower bound of $\Omega(\Delta' \log n)$ on the progress time of centralized algorithms in the dual graph model with collision detection. This lower bound directly yields a lower bound with the same value on the acknowledgment time in the same model. Together, these two bounds show that the optimized decay protocol presented in section \ref{sec:upper} achieves almost optimal acknowledgment and progress bounds in the dual graph model. On the other hand, this result demonstrates a big gap between the progress bound in the two models, proving that progress is unavoidably harder (slower) in the dual graph model.\fullOnly{ Also, we show an unavoidable big gap in the dual graph model between the receive bound, the time by which all neighbors of an active process have received its message, and the acknowledgment bound, the time by which this process believes that those neighbors have received its message.}
\fullOnly{\subsection{Lower Bound on the Progress Time}\label{subsec:Prog_Dual}

In the previous section, we proved a lower bound of $\Omega(\Delta \log n)$ for the acknowledgment time in the classical radio broadcast model. 
Now, we use that result to show a lower bound of $\Omega(\Delta' \log n)$ on the progress time in the dual graph model.}
\fullOnly{

To get there, we first need some definitions. Again, we will work with bipartite networks and in a one-shot setting. However, this time, these networks would be in the dual graph radio broadcast model and for each such network, we have two graphs $G$ and $G^\prime$. For each algorithm $A$ and each bipartite network in the dual graph model, we say that an execution $\alpha$ of $A$, is \emph{progressive} if throughout this execution, every receiver of that network receives at least one message. Note that an execution includes the choices of adversary about activating the unreliable links in each round. Now we are ready to see the main result of this section. 
}
\shortOnly{
\begin{theorem} \label{thm:worst-prog-dual} In the dual graph model, for each $n$ and each $\Delta' \in [20 \log n, n^{\frac{1}{11}}]$, there exists a bipartite network $H^*(n, \Delta')$ with $n$ nodes and maximum receiver $G^\prime$-degree at most $\Delta'$ such that no algorithm can have progress bound of $o(\Delta' \log n)$ rounds. In the same network, no algorithm can have acknowledgment bound of $o(\Delta' \log n)$ rounds.  
\end{theorem}
}
\fullOnly{
\begin{theorem} \label{thm:worst-prog-dual} In the dual graph model, for each $n_1$ and each $\Delta'_1 \in [20 \log n_1, n_1^{\frac{1}{11}}]$, there exists a bipartite network $H^*(n_1, \Delta'_1)$ with $n_1$ nodes and maximum receiver $G^\prime$-degree at most $\Delta'_1$ such that no algorithm can have progress bound of $o(\Delta' \log n_1)$ rounds.
\end{theorem}
}
\begin{proof}[Proof Outline]In order to prove this lower bound, in Lemma \ref{lem:trans}, we show a reduction from acknowledgment in the bipartite networks of the classical model to the progress in the bipartite networks of the dual graph model. In particular, this means that if there exists an algorithm with progress bound of $o(\Delta^\prime \log n)$ in the dual graph model, then for any bipartite network $H$ in the classical broadcast model, we have a transmission schedule $\sigma(H)$ with length $o(\Delta \log n)$ that covers $H$. Then, we use \Cref{thm:Ack_LB} to complete the lower bound.
\end{proof}

\begin{lemma}\label{lem:trans} Consider arbitrary $n_2$ and $\Delta_2$ and let $n_1 = n_2 \Delta_2$ and $\Delta'_1 = \Delta_2$. Suppose that in the dual graph model, for each bipartite network with $n_1$ nodes and maximum receiver $G'$-degree $\Delta'_1$, there exists a local broadcast algorithm $A$ with progress bound of at most $f(n_1, \Delta_1^\prime)$. Then, for each bipartite network $H$ with $n_2$ nodes and maximum receiver degree $\Delta_2$ in the classical radio broadcast model, there exists a transmission schedule $\sigma(H)$ with length at most $f({n_2}{\Delta_2}, \Delta_2)$ that covers $H$.
\end{lemma}  
\shortOnly{
\begin{proof}[Proof Sketch] Let $H$ be a network in the classical radio broadcast model with $n_2$ nodes and maximum receiver degree at most $\Delta_2$. We use algorithm $A$ to construct a transmission schedule $\sigma_H$ of length at most $f({n_2}{\Delta_2},\Delta_2)$ that covers $H$. We first construct a new bipartite network, \emph{Dual($H$)} = $(G, G')$, in the dual graph model with at most $n_1$ nodes and maximum receiver $G^\prime$-degree $\Delta'_1$. The set of sender nodes in the Dual($H$) is equal to that in $H$. For each receiver $u$ of $H$, let $d_{H}(u)$ be the degree of node $u$ in graph $H$. Let us call the senders that are adjacent to $u$ `the \emph{associates} of $u$'. In the network Dual($H$), we replace receiver $u$ with $d_{H}(u)$ receivers and we call these new receivers `the \emph{proxies} of $u$'. In graph $G$ of Dual($H$), we match proxies of $u$ with associates of $u$, i.e., we connect each proxy to exactly one associate and vice versa. In graph $G^\prime$ of Dual($H$), we connect all proxies of $u$ to all associates of $u$. It is easy to check that Dual($H$) has the desired size and maximum receiver degree.

Now we present a special adversary for the dual graph model. Later we construct transmission schedule $\sigma_H$ based on the behavior of algorithm $A$ in network Dual($H$) against this adversary. This special adversary activates the unreliable links using the following procedure. Consider round $r$ and receiver node $w$. (1) If exactly one $G^\prime$-neighbor of $w$ is transmitting, then the adversary activates only the links from $w$ to its $G$-neighbors, (2) otherwise, adversary activates all the links from $w$ to its $G^\prime$-neighbors.

We focus on the executions of algorithm $A$ on the network Dual($H$) against the above adversary. By assumption, there exists an execution $\alpha$ of $A$ with length at most $f(n_2 \Delta_2, \Delta_2)$ rounds such that in $\alpha$, every receiver receives at least one message. Let transmission schedule $\sigma_H$ be the transmission schedule of execution $\alpha$. Note that because of the above choice of adversary, in the execution $\alpha$, each receiver can receive messages only from its $G$-neighbors. Suppose that $w$ is a proxy of receiver $u$ of $H$. Then because of the construction of Dual($H$), each receiver node has exactly one $G$-neighbor and that neighbor is one of associates of $u$ (the one that is matched to $w$). Therefore, in execution $\alpha$, for each receiver $u$ of $H$, in union, the proxies of $u$ receive all the messages of associates of $u$. On the other hand, because of the choice of adversary, if in round $r$ of $\sigma$ a receiver $w$ receives a message, then using transmission schedule $\sigma_H$ in the classical radio broadcast model, $u$ receives the message of the same sender in round $r$ of $\sigma_H$. Therefore, using transmission schedule $\sigma_H$ in the classical broadcast model and in network $H$, every receiver receives messages of all of its associates. Hence, $\sigma_H$ covers $H$ and we are done with the proof of lemma.    
\end{proof}
}
\fullOnly{
\begin{proof} 
Consider an arbitrary $n_2$ and $\Delta_2$ and let $n_1 = n_2 \Delta_2$ and $\Delta'_1 = \Delta_2$. Suppose that in the dual graph model and for each bipartite network with $n_1$ nodes and maximum receiver $G'$-degree $\Delta'_1$, there exists a local broadcast algorithm $A$ for this network with progress bound of at most $f(n_1, \Delta_1^\prime)$. Let $H$ be a network in the classical radio broadcast model with $n_2$ nodes and maximum receiver degree at most $\Delta_2$. We show a transmission schedule $\sigma_H$ of length at most $f({n_2}{\Delta_2},\Delta_2)$ that covers $H$. 

For this, using network $H$, we first construct a special bipartite network in the dual graph model, \emph{Dual($H$)} = $(G, G')$ that has $n_1$ nodes and maximum receiver $G^\prime$-degree $\Delta'_1$. Then, by the above assumption, we know that there exists a local broadcast algorithm $A$ for this network with progress bound of at most $f(n_1, \Delta'_1) = f(n_2 \Delta_2, \Delta_2)$ rounds. We define transmission schedule $\sigma_H$ by emulating what this algorithm does on the network Dual($H$) and under certain choices of the adversary. Then, we argue why $\sigma_H$ covers $H$. 

The network Dual($H$) in the dual graph model is constructed as follows. The set of sender nodes in the Dual($H$) is exactly the same as those in $H$. Now for each receiver $u$ of $H$, let $d_{H}(u)$ be the degree of node $u$ in graph $H$. Also, let us call the senders that are adjacent to $u$ the \emph{associates} of $u$. Then, in the network Dual($H$), we replace receiver $u$ with $d_{H}(u)$ receivers and we call these new receivers the \emph{proxies} of $u$. Also, in graph $G$ of Dual($H$), we match proxies of $u$ with associates of $u$, i.e., we connect each proxy to exactly one associate and vice versa. In graph $G^\prime$ of Dual($H$), we connect all proxies of $u$ to all associates of $u$. Note that because of this construction, we have that the maximum degree of the receivers in $G^\prime$ is $\Delta_2$. Also, since each receiver is substituted by at most $\Delta_2$ receiver nodes, the total number of nodes mentioned so far in the Dual($H$) is at most $n_2 \Delta_2$. Without loss of generality, we can assume that the number of nodes in $Dual(H)$ is exactly $n_2 \Delta_2$. This is because we can simply adjust it by adding enough isolated senders. 

Now, we present a particular way of resolving the nondeterminism in the choices of adversary in activating the unreliable links for each round over Dual($H$). Later, we will study and emulate the algorithm $A$ under the assumption that the unreliable links are activated in this way. This method of resolving the nondeterminism is, in principle, trying to make the number of successful message deliveries as small as possible. More precisely, adversary activates the links using the following procedure. For each round $r$ and each receiver node $w$, we use these rules about the link activation: (1) if exactly one $G^\prime$-neighbor of $w$ is transmitting, then the adversary activates only the links from $w$ to its $G$-neighbors, (2) otherwise, adversary activates all the links from $w$ to its $G^\prime$-neighbors.

Now, we focus on the executions of algorithm $A$ on the network Dual($H$) and under the above method of resolving the nondeterminism. By the assumption that $A$ has progress time bound of $f(n_2 \Delta_2, \Delta_2)$ for network Dual($H$), there exists a progressive execution $\alpha$ of $A$ with length at most $f(n_2 \Delta_2, \Delta_2)$ rounds. Let transmission schedule $\sigma_H$ be the transmission schedule of execution $\alpha$. Note that in the execution $\alpha$, because of the way that we resolve collisions, each receiver can receive messages only from its $G$-neighbors. Suppose that $w$ is a proxy of receiver $u$ of $H$. Then because of the construction of Dual($H$), each receiver node has exactly one $G$-neighbor and that neighbor is one of associates of $u$ (the one that is matched to $w$). Therefore, in execution $\alpha$, for each receiver $u$ of $H$, in union, the proxies of $u$ receive all the messages of associates of $u$. Now, note that because of the presented method of resolving the nondeterminism, if in round $r$ of $\sigma$, a receiver $w$ receives a message, then using transmission schedule $\sigma_H$ in the classical radio broadcast model, $u$ receives the message of the same sender in round $r$ of $\sigma_H$. Therefore, using transmission schedule $\sigma_H$ in the classical broadcast model and in network $H$, every receiver receives messages of all of its associates. Hence, $\sigma_H$ covers $H$ and we are done with the proof of lemma.    
\end{proof}
\begin{proof}[Proof of \Cref{thm:worst-prog-dual}] The proof follows from Theorem \ref{thm:Ack_LB} and Lemma \ref{lem:trans}. Fix an arbitrary $n_1$ and $\Delta'_1 \in 20 \log n_1, n_1^{\frac{1}{11}}]$. Let $n_2 = \frac{n_1}{\Delta'_1}$ and $\Delta_2 = \Delta'_1$. By theorem \ref{thm:Ack_LB}, we know that in the classical radio broadcast model, there exists a bipartite network $H(n_2, \Delta_2)$ with $n_2$ nodes and maximum receiver degree at most $\Delta_2$ such that no transmission schedule with length of $o(\Delta_2 \log n_2)$ rounds can cover it. Then, by setting $f(n_1, \Delta_1) = \Theta (\Delta_1 \log n_1)$ in Lemma \ref{lem:trans}, we can conclude that there exists a bipartite network with $n_1$ nodes and maximum receiver $G'$-degree $\Delta'_1$ such that there does not exists a local broadcast algorithm for this network with progress bound of at most $f(n_1, \Delta_1^\prime)$. Calling this network $H^*(n_1, \Delta'_1)$ finishes the proof of this lemma.

\end{proof}
}
\fullOnly{\begin{corollary} In the dual graph model, for each $n$ and each $\Delta' \in [20 \log n, \frac{n^{\frac{1}{11}}}{2}]$, there exists a bipartite network with $n$ nodes and maximum receiver $G^\prime$-degree at most $\Delta'$ such that for every $k \in [20 \log n, \Delta']$, no algorithm can have progress bound of $f_{prog}(k) = o(k \log n)$ rounds. 
\end{corollary}
\begin{proof} The corollary follows from \ref{thm:worst-prog-dual} by considering the dual network graph that is derived from union of networks $H^*(\frac{n}{2}, k)$ as $k$ goes from $20 \log n$ to $\Delta'$.
\end{proof}

\begin{corollary}In the dual graph model, for each $n$ and each $\Delta' \in [20 \log n, \frac{n^{\frac{1}{11}}}{2}]$, there exists a bipartite network $H^*(n, \Delta')$ with $n$ nodes and maximum receiver $G^\prime$-degree at most $\Delta'$ such that no algorithm can have acknowledgment bound of $o(\Delta' \log n)$ rounds. 
\end{corollary}
\begin{proof} Proof follows immediately from \ref{thm:worst-prog-dual} and the fact that the acknowledgment time is greater than or equal to the progress time.
\end{proof}
}
\fullOnly{\subsection{The Intrinsic Gap Between the Receive and Acknowledgment Time Bounds}\label{subsec:rcv_ack_gap}
In \Cref{sec:upper}, we saw that the SPP protocol has a reception time bound of $f_{rcv}(k) = O(k \log(\Delta') \log n)$. In this section, we show that in the distributed setting, there is a relatively large gap between the time that the messages can be delivered in and the time needed for acknowledging them. More formally, we show the following.
\begin{lemma}In the dual graph model, for each $n_1$ and each $\Delta' \in [20 \log n_1, \frac{n_1^{\frac{1}{11}}}{2}]$, there exists a bipartite network $\mathcal{H}_{rcv}(n_1, \Delta'_1)$ with $n_1$ nodes and maximum receiver $G^\prime$-degree at most $\Delta'$ such that for any distributed algorithm, many senders have $c'(v, r) \leq 1$, but they can not acknowledge their packets in $o(\Delta'_1 \log n_1)$ rounds, i.e., $ f_{ack}(1) = \Omega (\Delta'_1 \log n_1)$.
\end{lemma}

\begin{proof} Let $n_2= \lfloor \frac{n_1^{\frac{10}{11}}}{2}\rfloor$ and $\Delta_2 = \Delta'_1$. Then, let $H( n_2, \Delta_2)$  with size $n_2$ and maximum degree $\Delta_2$ be the bipartite network in the classic model that we showed its existence in \Cref{thm:Ack_LB}. Recall that in $H(n_2, \Delta_2)$, we have $\eta = (n_2)^{\,0.1}$ sender processes. Now, we first introduce two simple graphs using $H(n_2, \Delta_2)$. Add $\eta$ receivers to the receiver side of $H(n_2, \Delta_2)$, call them \emph{new receivers}, and match these new receivers to the senders. Let us call the matching graph itself $M$. Then, define $G' = H(n_2, \Delta_2)+M$, $G_1 = M$ and $G_2 = H(n_2, \Delta_2)+M$. Also, let $\mathcal{H}_{rcv}(n_1, \Delta'_1)$ be the dual graph network that is composed of two components, one being the pair $(G_1, G')$ and the other being $(G_2,G')$. In each pair, the first element is the reliable part of the component and the second is the whole component. Note that the total number of nodes in $\mathcal{H}(n_1, \Delta'_1)$ is at most $n_1^{\frac{10}{11}} + n_1^{\frac{1}{11}}$ which is less than or equal to $n_1$ for large enough $n$. Without loss of generality, we can assume the number of nodes in $\mathcal{H}(n_1, \Delta'_1)$ is exactly $n_1$ by adding enough isolated nodes. 

Now note that the second component of $\mathcal{H}_{rcv}(n_1, \Delta'_1)$, which is the pair $(G_2,G')$, $G'$ is a super graph of $H(n_2, \Delta_2)$. Hence, Lemma \ref{lem:Ack_Schedules_LB}, for any algorithm, acknowledgment in the second component needs at least $\Omega(\Delta_2 \log(n_2)) = \Omega(\Delta'_1 \log n_1)$ rounds. On the other hand, since for every new receiver $u$ in the first component, we have $|\mathcal{N}_{G_1}(u)| = 1$, we know that for every sender $v$ in the first component, for any round $r$ of any algorithm, $c'(v, r) \leq 1$. Now consider an arbitrary subset $P$ of all processes with $|P|=\eta'$. As an adversary, we can map these processes into either the senders in the first component or the senders in the second component. Since processes don't know the mapping between the processes, if we resolve the nondeterminism by always activating all the edges, the processes can not distinguish between the aforementioned two cases of mapping. Hence, since acknowledgment in the second component takes at least $\Omega(\Delta'_1 \log n_1)$ rounds, it takes at least the same amount of time in the first component as well. Thus, this dual graph network satisfies all the desired properties for $\mathcal{H}_{rcv}(n_1, \Delta'_1)$ mentioned in the theorem statement and therefore, we are done with the proof.
\end{proof}
}


%% file: Gap.tex
\section{Centralized vs. Distributed Algorithms in the Dual Graph Model}

%

In this section, we show that there is a gap in power between
distributed and centralized algorithms in the dual graph model,
but not in the classical model---therefore highlighting
another difference between these two settings.
Specifically, we produce dual graph network graphs where centralized
algorithms achieve $O(1)$ progress while 
distributed algorithms have unavoidable slow progress.
In more detail,
our first result shows that distributed algorithms
will have {\em at least one process}
experience 
$\Omega(\Delta'\log{n})$
progress,
while the second result shows
the {\em average} progress is
$\Omega(\Delta')$.
Notice, such gaps do not exist in the classical model,
where our distributed algorithms from Section~\ref{sec:upper}
can guarantee fast progress in all networks.

\begin{theorem}\label{thm:worst-prog-gap}
For any $k$ and $\Delta'\in [20\log{k},k^{1/10}]$,
there exists a dual graph network  of size $n$, $k < n \leq k^4$,
with maximum receiver degree $\Delta'$,
such that the optimal centralized local broadcast
algorithm achieves a progress bound of $O(1)$ in this network
while every distributed
local broadcast algorithm has a progress bound of 
$\Omega(\Delta'\log{n})$.
\end{theorem}
\shortOnly{Our proof argument leverages
the bipartite network proven to exist in Lemma~\ref{lem:trans} to show that
all algorithms have slow progress in the dual graph model.
Here, we construct a network consisting of many copies
of this counter-example graph. In each copy, we leave
one of the reliable edges as reliable, but {\em downgrade}
the others to unreliable edges that act reliable. 
A centralized algorithm can achieve fast progress in each
of these copies as it only needs the processes connected
to the single reliable edge to broadcast.
A distributed algorithm, however, does not know
which edge is actually reliable, so it still has slow
progress. We prove that in one of these copies, the last
message to be delivered comes across the only reliable edge,
w.h.p. This is the copy that provides the slow progress needed
by the theorem.}
\fullOnly{\begin{proof}
Let $G_1 = H(\Delta',n)$ be the classic network,
with size $n$ and maximum receiver degree $\Delta'$,
proved to exist by Theorem~\ref{thm:Ack_LB}.
(Notice the bounds on $\Delta'$ from the theorem
statement match the requirement by Theorem~\ref{thm:Ack_LB}.)
As also proved in this previous theorem, 
every centralized
algorithm has an acknowledgment bound 
of $\Omega(\Delta'\log{n} )$
in $G_1$.

Next, let $G_2 = Dual(G_1)$ be the dual graph
network, with maximum receiver degree $\Delta'$ 
and network size $n_2 = n\Delta'$,
that results from applying the $Dual$ transformation,
defined in the proof of Lemma~\ref{lem:trans}, to $G_1$. 
This Lemma proves that every centralized
algorithm has a progress bound
 of $\Omega(\Delta'\log{n_2})$
rounds in $G_2$.
We can restate this bound as follows:
for every algorithm, there is an assignment
of messages to senders such
that in every execution
some process has a reliable
edge to at least one sender,
and yet does not receive its first message from
a sender for $\Omega(\Delta'\log{n_2})$
rounds. Call the reliable edge on which this slow process
receives its first message the {\em slow edge} in the execution.\footnote{We
are assuming w.l.o.g. that in these worst case executions
identified by the lower bound, that the last receiver to receive
a message does not receive this message on an unreliable edge
(as, in this case, we could always drop that message, contradicting
the assumption that we are considering the worst case execution).}

We now use $G_2$ to construct a larger dual graph network, $G^*$.
To do so, label the $m$ reliable edges in $G_2$ as $e_1,...,e_m$.
We construct $G^*$  to consist of $n_2m^2$ modified copies of 
$G_2$. 
In more detail, $G^*$ has $n_2m$ components, 
which we label $C_{i,j}$, $i\in [m], j\in [n_2m]$.
Each $C_{i,j}$ has the same structure
as  $G_2$ but with the following exception:
we keep only $r_i$ as a reliable edge;
all other reliable edges $r_j$, $j\neq i$,
are {\em downgraded} to unreliable edges.

We are now ready to prove a lower bound on progress 
on $G^*$.
Fix some distributed local broadcast algorithm ${\cal A}$.
We assign the $n_2^2m^2$ process to nodes in $G^*$
as follows.
Partition these processes into sets
$S_1,...,S_{n_2m^2}$, each consisting of $n_2$ processes.
For each $S_i$, $i\in [n_2m]$,
we make an independent random choice of a value $j$ from $[m]$,
and assign $S_i$ to component $C_{j,i}$ in $G^{*}$.
Notice, no two such sets can be assigned to same to the same
component, so the choice of each assignment can be independent
of the choice of other assignments. We also emphasize
that these choices are made independent of the algorithm 
${\cal A}$ and its process' randomness.
Finally, we assign the remaining $S$ sets to the
remaining $G^*$ components in an arbitrary fashion.

For each $C_{j,i}$,
we fix the behavior of each downgraded edge to
behave as if it was a reliable edge.
With this restriction in place, $C_{j,i}$ now
behaves indistinguishably from $G_2$.
It follows from Lemma~\ref{lem:trans},
that no algorithm can guarantee fast
progress in $C_{j,i}$.

Leveraging this insight, we assume
the worst case behavior,
in terms of the non-downgraded unreliable
edge behavior and message assignments,
in each component.
In every $C_{j,i}$, therefore,  
some process does not receive a message for the first
time on a reliable or downgraded edge
for $\Omega(\Delta'\log{n_2} )$ rounds.
With this in mind, let us focus
on our sets of processes $S_1$ to $S_{n_2m^2}$.
Consider some $S_i$ from among
these sets. Let $C_{j,i}$ be
the component to which we randomly assigned $S_i$.
As we just established, 
some process in $S_i$ does
not receive a message for the first
time until many rounds have passed.
This message either comes across
the single reliable edge in $C_{j,i}$
or a downgraded edge.
If it is a reliable edge, 
then this process yields the slow progress we need.

The crucial observation here is that 
for any fixed randomness for the processes
in $S_i$,
the choice of this edge is the same regardless
of the component where $S_i$ is assigned.
Therefore we can treat the determination
of this slow edge as independent of
the assignment of $S_i$ to a component.
Because we assigned $S_i$ at random
to a component, the probability
that we assigned it to a component
where the single reliable edge matches
the fixed slow edge is $1/m$.
Therefore, the probability
that this match occurs for
at least one of our $n_2m$ $S$
sets is $(1 - 1/m)^{n_2m} \leq 1 - e^{n_2}$.
In other words,
some receiver in our network does not
receive a message over a reliable edge
for a long time, w.h.p.
Because a progress bound must hold w.h.p.,
the progress bound of ${\cal A}$ is slow.

Finally, to establish our gap, we must also describe
a centralized algorithm can achieve $O(1)$
progress in this same network, $G^*$.
To do so, notice each component $C_{i,j}$ 
has exactly one reliable edge.
With this in mind, we define our fixed centralized algorithm to divide rounds
in pairs and do the following: in the first round of a pair,
if the first endpoint of a component's single reliable edge (by some
arbitrary ordering of endpoints)
has a message then it broadcasts; in the second
round do the following for the second endpoint.
After a process has been active for a full round pair, it acknowledges the message.
This centralized algorithm satisfies the following property:
if some process $u$ receives a messages as input in round $r$, 
every reliable neighbor of $u$ receives the message by $r+O(1)$.
It follows that this centralized algorithm
has a progress bound of $O(1)$.
\end{proof}
}

Notice, in some settings, practioners might tolerate a slow worst-case
progress (e.g., as established in Theorem~\ref{thm:worst-prog-gap}),
so long as {\em most} processes have fast progress.
In our next theorem, we show that this ambition is also impossible
to achieve.
To do so, we first need a definition that captures
the intuitive notion of many processes having slow progress.
In more detail, given an execution of the one-shot local broadcast
problem (see Section~\ref{sec:model}), with 
processes in {\em sender set} $S$ being passed messages,
label each receiver that neighbors
$S$ in $G$ with the round when it
first received a message. The {\em average progress}
of this execution is the average of these values.
We say an algorithm has an {\em average progress of $f(n)$},
with respect to a network of size $n$ and sender set $S$, 
if executing
that algorithm in that network with those senders
generates an
average progress value of no more than $f(n)$, w.h.p.
We now bound this metric in the same style as above

\begin{theorem}\label{thm:average-prog-gap}
For any $n$, there exists a dual graph network of size $n$
and a sender set,
such that the optimal centralized local broadcast
algorithm has an average progress of $O(1)$
while every distributed local broadcast algorithm
has an average progress of $\Omega(\Delta')$.
\end{theorem}
\shortOnly{Our proof uses a reduction
argument. We show how a distributed algorithm that achieves
fast average progress in a specific type of dual graph network
can be transformed to a distributed algorithm that
solves global broadcast fast in a different type
of dual graph network.
We then apply a lower bound from~\cite{KLN09BA}
that proves no fast solution exists for the latter---providing
our needed bound on progress.}
\fullOnly{
\paragraph{Lollipop Network.}
We begin our argument
by recalling a result proved in a previous study of
the dual graph model.
This result concerns the {\em broadcast
problem}, in which a single source process is provided
a message at the beginning of the execution which it must
subsequently
propagate to all processes in the network.
The result in question
 concerned a specific dual graph construction
we call a {\em lollipop network}, which can be defined
with respect to any network size $n>2$.
For a given $n$,
the $G$ edges in this network
define a clique of $n-1$ nodes,
$c_1$ to $c_{n-1}$.
There is an additional node
$r$ that is connected
to one of the clique clique nodes.
By contrast, $G'$ is complete.
In~\cite{KLN09BA} we proved the following:

\begin{lemma}[From~\cite{KLN09BA}]
Fix some $n>2$ and randomized broadcast algorithm ${\cal A_B}$.
With probability at least $1/2$,
${\cal A_B}$ requires at least $\lfloor (n-1)/2 \rfloor$
rounds to solve broadcast in the lollipop network of size $n$.
\label{lem:podc2009}
\end{lemma}

\paragraph{Spread Network.}
Our strategy in proving Theorem~\ref{thm:average-prog-gap}
is to build a dual graph network
in which achieving fast average progress would
yield a fast solution to the broadcast problem in the lollipop
network, contradicting Lemma~\ref{lem:podc2009}.
To do so, we need to define the network in which we achieve
our slow average progress.
We call this network a {\em spread network}, 
and define it as follows.
Fix any even size $n \geq 2$.
Partition the $n$ nodes in $V$ into 
{\em broadcasters} ($b_1,b_2,...,b_{n/2}$)
and {\em receivers} ($r_1,r_2,...,r_{n/2}$).
%
For each $b_i$, add a $G$ edge to $r_i$.
Also add a $G$ edge from
$b_1$ to all other receivers.
Define $G'$ to be complete. Note that in this network, $\Delta'=n-1$.


We can now prove our main theorem.


\begin{proof}[Proof of Theorem~\ref{thm:average-prog-gap}]
Fix our sender set $S = \{b_1,...,b_{n/2}\}$.
Notice, a centralized algorithm can achieve
$1$ round progress for all receiver by simply
have $b_1$ broadcast alone.

We now turn our attention to showing
that any distributed algorithm, by contrast, is slow
in this setting.
Fix one such algorithm, ${\cal A}$.
Assume for contradiction that it defies
the theorem statement.
In particular, it will guarantee $o(n)$ progress
when executed in the spread network
with sender set $S = \{b_1,...,b_{n/2}\}$.

%
We use  ${\cal A}$
to construct a broadcast algorithm ${\cal A'}$
that can be used to solve broadcast in the lollipop network.
At a high-level, ${\cal A'}$ has each
process in the clique in the lollipop network simulate both
a sender and its matching receiver from the spread network.
In the following, use $b$ to refer to the single node in
the clique of the lollipop network that connects to $r$ with
a reliable edge.
In this simulation, process $b$ in the lollipop network matches up with
process $b_1$ in the spread graph.
Of course, process $b$ does not know a priori that it is 
simulating process $b_1$, as in the lollipop network
$b$ does not a priori that is assigned to this crucial node.
This will not be a problem, however, because 
we will control the $G'$ edges in our simulation
such that the behavior of $b_1$ will differ from
the other processes in $S$ only when it broadcasts alone
in the graph. It will be exactly at this point,
however, that our simulation can stop, having
successfully solved broadcast.

In more detail, our algorithm ${\cal A'}$ works as follows:

\begin{enumerate}

  \item We first allow process $r$ to identify itself.
  To do so, have the source, $u_0$, broadcast.
  Either we solve the broadcast problem (e.g., if the source is $b$)
  or $r$ is the only process to not receive a message---allowing
  it to figure out it is $r$. At this point, every
   process but $r$ has the message. To solve broadcast
  going forward, it is now sufficient for $b$ to broadcast alone.

  \item We will now have processes in ${\cal A'}$ simulate
  processes from ${\cal A}$ to determine whether or
  not to broadcast in a given round.
  In more detail, we have each process $u$ in the lollipop
   clique simulate a sender (call it, $b_u$) and its corresponding
   receiver (call it, $r_u$) from the spread network.\footnote{In
   the case of the process simulating $b_1$, we have to be careful
   because $b_1$ has a $G$ edge to all receivers. The
  simulator, however, is responsible only for simulating
   the sole receiver that is connected to only $b_1$, namely $r_1$.}
  We have $n/2$ clique processes each simulating $2$ spread network
  processes, so we are now setup to begin a  simulation of
  an $n$-process spread network. 

 \item Each simulated round of ${\cal A}$ will require
       two real rounds of ${\cal A'}$.

      {\em In the first real round}, each process $u$ in the lollipop
      clique advances the simulation of its simulated
     processes $b_u$ and $r_u$, to see if they
      should broadcast in the current
     round of ${cal A}$ being simulated. If either $b_u$ or $r_u$
      broadcasts (according to $u$'s simulation), 
       $u$ broadcasts these simulated messages,
      {\em and} the broadcast
     message for the instance of broadcast we are trying to solve.
    On the other hand, if neither of $u$'s simulated processes broadcast,
   $u$ remains silent.
     (Notice, if only $b$ broadcasts during this round, we are done.)

    The exception to these rules is the source, $u_0$, which does not broadcast,
    regardless of the result of its simulation.

    {\em In the second real round of our simulated round}, 
    $u_0$ announces what it learned in the previous round.
    That is, $u_0$ acts as a simulation coordinator. 

     In more detail, 
    $u_0$ can tell the difference between the following
     two cases:
     (1) either no simulated process, or two or more
         simulated processes, broadcast;
     (2) one simulated process broadcast (in which
         case $u_0$ also knows whether the processes is a sender
         or receiver in the spread network, and its
          message);

    Process $u_0$ announces whether case $1$ or $2$ occurred, and in the
     case of (2), it also announces the identity of the
     sender and its message.
    This information is received by all processes in the lollipop clique.

\item Once the lollipop clique processes learn 
      the result of the simulation
     from $u_0$, they can consistently and correctly finish
     the round for their simulated processes by applying
 the following rules.

  {\em Rule \#1:} If $u_0$ announces that no simulated process
  broadcasts, or two or more simulated processes broadcast,
   then all the processes in ${\cal A'}$ 
  have their simulated processes receive nothing.
  (This is valid as $G'$ is complete in the simulated network,
   so it is valid for concurrent messages
   to lead to total message loss.)
 

  {\em Rule \#3:} If $u_0$ announces that one simulated process broadcast,
   then the simulators' behavior depends on the identity of the
   simulated broadcaster. If this broadcaster is a sender in the 
  spread network, then it simulates its single matched receiver receiving
 the message. (Notice this behavior is valid so long as the broadcaster
   is not $b^*$. Fortunately, the broadcaster {\em cannot} be $b^*$, 
   as if it was, then
   $b$ would have broadcast alone in ${\cal A'}$ in the previous round,
   solving broadcast.)
  
   On the other hand, if the single broadcaster is a receiver,
   then we have to be more careful. It is not sufficient for
its single matched broadcaster to receive the message
  because $b_1$ must also receiver it.
  Because we do not know which process is simulating $b_1$,
  we instead, in this case, simulate all broadcasters
  receiving this message. This is valid as $G'$ is complete.

\end{enumerate}

By construction, ${\cal A'}$ will solve broadcast when simulated
$b_1$ broadcasts
alone in the simulation. 
Our simulation rules are designed such that $b_1$ {\em must} eventually
broadcast alone for the simulated instance of ${\cal A}$ to solve
local broadcast, as this is the only way for $r_1$
to receive a message from a process in $S$.
Because we assume ${\cal A}$ solves this problem,
and we proved our simulation of ${\cal A}$ is valid, 
${\cal A}$ {\em will} eventually have $b_1$ broadcast alone and
therefore ${\cal A'}$ {\em will} eventually solve broadcast.

The question is how long it takes for this event to occur.
Recall that we assumed that with high probability
the average progress of ${\cal A}$ is $o(n)$.
By our simulation rules, until $b_1$ broadcasts alone,
at most one receiver can receive a message from 
a sender, per round. It follows that $b_1$ must broadcast
alone (well) {\em before} round $n/4$. (If it waited until
$n/4$, only $n/4$ processes will have finished receiving
in those round, so even if the remaining receivers all finished in round
$n/4$, the average progress would be greater
than $n/8$ which, of course, is not $o(n)$.)

By Lemma~\ref{lem:podc2009}, 
with probability at least $1/2$,
${\cal A'}$ requires at least
 $((\frac{n}{2}+1) - 1)/2 = n/4$ rounds to solve broadcast.
We just argued, however, that with {\em high} probability 
$b_1$ broadcasts alone---and therefore ${\cal A'}$ solves
broadcast---in less than $n/4$ rounds. A contradiction.
\end{proof}
}

%% file: Ref-List.tex
\bibliographystyle{plain}